\documentclass[12pt]{amsart}
\usepackage{latexsym,fancyhdr,amssymb,color,amsmath,amsthm,graphicx,listings,comment}
\newtheorem{thm}{Theorem} \newtheorem{lemma}{Lemma} 
 
\let\paragraph\subsection

\title{On Particles and Primes}
\author{Oliver Knill}
\date{Aug 21, 2016}
\address{Department of Mathematics \\ Harvard University \\ Cambridge, MA, 02138 }

\begin{document}
\maketitle

\begin{abstract}
Primes in the two complete associative normed division algebras 
$\mathbb{C}$ and $\mathbb{H}$ have affinities with structures seen in the standard model 
of particle physics. On the integers in the two algebras, there are two equivalence 
relations: a strong one, related to a $U(1)$ and $SU(3)$ symmetry
allowing to permute and switch signs of the coordinates of the integers,
as well as a weak relation with origins from units $U(1),SU(2)$ in 
the algebra. Weak equivalence classes within the strong equivalence classes of 
odd primes in $\mathbb{C}$ case relate to leptons, the inert ones being neutrino like, 
and the split ones resembling electron-positron pairs. In the $\mathbb{H}$ case, for odd primes, the 
equivalence classes come in groups of two or three, leading to a caricature of hadrons 
featuring either mesons built by a quark pair or then baryons obtained by quark triplets.
We can now list for every rational prime $p$ all these particles and attach fractional 
charges to its constituents. 
\end{abstract}

\section{Introduction}

\paragraph{} 
When experimenting with primes in division algebras, an
affinity of primes in $\mathbb{C}$ and $\mathbb{H}$ with the structure in the 
standard model of particle physics emerged: 
primes in $\mathbb{C}$ resemble leptons, while equivalence classes of 
primes in $\mathbb{H}$ have Hadron like features \cite{Experiments}.
In $\mathbb{C}$, where besides the ramified Gaussian primes of norm $2$ like $1+i$,
two other type of primes exist, there is a lepton structure:
in $\mathbb{C}$ we see cartoon versions of neutrini as well as electron-positron pairs. 
In $\mathbb{H}$, where besides the ramified primes above $p=2$, only one type of primes 
exists, modulo $U(3)$ gauge transformations, the structure of $SU(2)$-equivalence classes appear
to be `hadrons": there are either mesons or baryons, where the individual elements of the 
equivalence classes resemble quarks. 

\paragraph{} 
The algebras $\mathbb{C},\mathbb{H}$ are naturally distinguished as these two are the only 
complete associative normed division algebras. Not only the unit spheres $U(1),SU(2)$ of
in $\mathbb{C},\mathbb{H}$ but also the Lie group $SU(3)$ acts on spheres in $\mathbb{H}$ 
by diffeomorphisms rotating three complex planes hinged together at the real axes. This non-linear
action allows to implement more symmetries or extend any quantum dynamics from $\mathbb{C}$ 
to $\mathbb{H}$ valued fields. These symmetries lead to finite norm preserving group 
actions on quaternion integers and can be implemented 
without mentioning Lie groups: the weak equivalence is obtained by using multiplication by units 
in $\mathbb{H}$; the strong symmetries on integers $(a,b,c,d)$ in $\mathbb{H}$ 
is generated by the 24 permutations 
or sign changes. Obviously this can be realized using both rotations or reflections in $U(1)$ 
or then by permutations and sign changes in the space components $(b,c,d)$ of the quaternion.
The upshot is that every integer has a strong equivalence class $(a,b,c,d)$ with non-negative 
entries where $a \leq b \leq c \leq d$. If we look at weak equivalence classes in such
integer classes, the result is that each class has either one, two or three elements, where
the case with one element only happens if $N(a,b,c,d)=2$. 

\paragraph{} 
The fact that the major gauge groups of the standard model appear as
symmetries in $\mathbb{H}$ suggests to look at {\bf quaternion quantum mechanics}, 
where waves are $\mathbb{H}$-valued. 
Such a physics could be realized by looking at quaternion valued wave
equations developed by Fueter \cite{Fueter34} who also noticed that the Dirac equation
in quaternion analysis plays the role of the Cauchy-Riemann differential equations. 
Alternatively, one can extend linear or nonlinear wave evolutions 
from $\mathbb{C}$-valued fields in sub planes to quaternion-valued fields: 
just evolve simultaneously the classical waves on planes spanned by $1$ and 
one of the spacial units like $i,j,k$ in $\mathbb{H}$. 
If $L$ is any self-adjoint Hamiltonian operator, then instead of 
quantum dynamics $\psi(t) = \exp(i L t) \psi$ for a wave $\psi=(a/\sqrt{3},b)$,
one can consider two other waves $\phi=(a/\sqrt{3},c),\theta=(a/\sqrt{3},d)$ and evolve
all simultaneously using $(\psi,\phi,\theta)' = i(L\psi,L\phi,L\theta)$. This assures that
$C=|\psi(t)|^2+|\phi(t)|^2+|\theta(t)|^2$ stays invariant. The three evolutions can be merged
together to a quaternion with this norm $C=N(a,b,c,d)$ by
putting $b(t)={\rm Im}(\psi(t)),c(t)={\rm Im}(\phi(t)),d(t)={\rm Im}(\theta(t))$ and 
$a(t)= \sigma |(\psi,\phi,\theta)|$ with sign $\sigma$ chosen so that
$(a,b,c,d) \to (a(t),b(t),c(t),d(t))$ has positive Jacobean. This assures that 
$N(a,b,c,d)=a^2+b^2+c^2+d^2$ is time invariant. The second equivalence relation, 
the $SU(2)$ symmetry comes from the Cayley-Dickson picture expressing 
$(a,b,c,d)$ as a pair of two complex numbers $z=a+ib,w=c+id$ for which 
$|z|^2+|w|^2 = N(a,b,c,d) = {\rm det}(A)$ with the complex matrix $A$
built by Pauli matrices. 
These weak symmetries on $(z,w)$ can be implemented as determinant preserving unitary 
$2 \times 2$ matrices and so by elements in $SU(2)$. 
Obviously, the weak and strong equivalence classes are different. 
Looking at both together in $\mathbb{H}$ gives the Hadron structure on primes.
This amusing fact is is pure finite combinatorics.

\paragraph{} 
One can ask now why number theory should matter in quantum dynamics, 
as a random quaternion is almost surely never an integer. 
One can look however at the dynamics on a 
``geometry", a space $X$ with exterior derivative and measure $\mu$, 
where one classically evolves waves in $L^2(X,\mu,\mathbb{C})$,
moving according to some dynamics in the Hilbert space $L^2(X,\mu,\mathbb{H})$ with norm 
$||\Psi||_2 = \int_X N(\Psi(x)) \; d\mu(x)$. 
The quaternion quantum evolution preserves this norm but at individual points $x$, 
the wave amplitude changes. Gauge symmetry now can render integers relevant: 
if a wave amplitude reaches integer arithmetic $\mathbb{H}$ norm $N(\Psi(t))$ 
at some point, the wave value $\Psi(t)$ can be gauged within the gauge groups
to become $\mathbb{H}$-integer valued at this point. 
Its prime factorization structure features now particles at this point. 
The switch of factorizations could now be seen as particle processes in which Fermions 
= odd primes exchange vector gauge bosons = units or ramified primes = neutral bosons. 
Unit migration is an exchange of bosons and recombination is particle pair 
creation or annihilation, mending in some way also the particle-wave duality conundrum.

\paragraph{} 
This picture most certainly is just an allegory or caricature even in the realm of particle 
phenomenology alone when disregarding dynamics: the reason is that even the simplest processes 
like {\bf beta decay of a neutron} $n \to p+e^-+\mu_e$ to a proton, 
electron, electron neutrino, or a {\bf pion decay} $\pi^+  = e^+ + \mu_e$, 
where a meson decays into two leptons can not
be explained by unit migration, recombination or meta commutation in $\mathbb{H}$ alone. Also, the
way, how charge is defined below, no baryons do exist for which all three quarks have charge $2/3$. 
This violates the existence of the $\Delta^{++}$ baryon. 
We believe however that the definition of charge given here can be modified to incorporate 
this. We have just chosen an algorithm (using Lipschitz or Hurwitz primes) 
which gives us a deterministic charge. There is no other reason why 
$(2/3,2/3,2/3)$ charge triplets are excluded. \\

\paragraph{} 
The picture drawn here is probably of little value for physics as the later by definition 
requires to be able to do {\bf quantitative predictions} or {\bf quantitative verifications} 
of experimentally observable processes. And the picture drawn here does neither.
We feel however that the story is of mathematical interest and that 
it motivates to look more closely at higher arithmetic in $\mathbb{H}$ as well at
the mathematical structure of standard model which appears to suffer from a
{\bf lack of inevitability}. While the standard model is one of the most successful physical
theories with excellent match between experiment and model, the structure of the involved
Lagrangians is complicated even if one looks at it from a non-commutative geometry point
of view. 

\paragraph{} 
Affinities between mathematical structures and physical phenomena are sometimes useful,
sometimes just amusing. Here is an example of the more amusing type \cite{Golomb}:
when looking at the Rubik group, a finite group, there are permutations
available when disassembling and recombining the actual cube which resemble quarks: turn a corner
cube by 120 degrees for example. Rotating one corner and other corner in the opposite direction
gives a quark-anti quark meson and this permutation is realizable. 
Turning three corners by 120 degree can physically be realized and 
resembles a baryon. Also this picture is useless for physics but it is of some value as it 
helps to learn more about the structure of that particular finite group. 

\paragraph{} 
The standard model does not answer why the number of generations of hadrons or leptons is 
limited to three or why the gauge groups are not unified to a larger group like $SU(5)$ predicted by some
grand unified theories. We have seen that $\mathbb{H}$ naturally features the gauge groups of
the standard model, even so the $SU(3)$ symmetry acts only by diffeomorphisms. 
The three generations can emerge naturally in a quaternion-valued quantum mechanics as we 
have to evolve waves in three different planes. Since the time scales in the different planes
are in general different (as there is no reason why they should agree), 
a particle involving in a faster plane appears to be lighter. Better than {\bf imposing} a wave
evolution is to let the system just evolve freely in its isospectral set. Wave dynamics then 
emerges naturally \cite{IsospectralDirac2} in the form of isospectral
Lax deformations of exterior derivatives leading to three generations of geometries.
As both leptons and hadrons move in at least one plane and more likely in 2 or 3 planes 
at the same time, the probability of having a light particles is larger. 
The heavier particles like the top quark are more rare. 
The neutrini in the intersection of these planes participate to any of three dynamics producing
thought associations with the observed phenomenon of neutrino oscillations. 

\paragraph{} 
Number theory in $\mathbb{H}$ has started with Hurwitz \cite{Hurwitz1922} who built the setup and
established some factorization features. The fundamental theorem of arithmetic in $\mathbb{H}$ was 
completed \cite{ConwaySmith}. The fact that quaternions were an outcast in part of the 20th century
can historically be traced to the success of vector calculus, especially as formulated by Gibbs and
Wilson \cite{GibbsWilson}, a book so successful that its content not only structured practically 
all modern calculus textbooks but also removed quaternions from the curricula. 
Quaternion calculus was still cultivated, like by the Swiss number theorist Rudolf Fueter who 
found a Cauchy integral theorem for quaternions (see \cite{Fueter34,Fueter35,Deavours73}). 
There is certainly much still to be explored. Only recently, the permutation structure of 
the meta-commutation in the prime factorization has been studied for the first time \cite{CohnKumar}. 

\section{The quaternion algebra}

\paragraph{}  
By the {\bf Hurwitz theorem} \cite{Hurwitz1922},
the algebras $\mathbb{R},\mathbb{C},\mathbb{H},\mathbb{O}$ form 
a complete list of all {\bf normed division algebras}. 
Related to the Hurwitz theorem is the {\bf Frobenius theorem}
\cite{Frobenius1879} which tells that the only associative real division algebras 
are $\mathbb{R},\mathbb{C},\mathbb{H}$ and the {\bf Mazur theorem} which assures 
that the only {\bf commutative Banach division algebras} are $\mathbb{R}$ and $\mathbb{C}$.
The associative, algebraically complete division algebras are $\mathbb{C}$ and 
$\mathbb{H}$. 
{\bf Algebraic completeness} in $\mathbb{H}$ means that every polynomial 
$\sum_i a_i x^i=0$ has a solution $x$ \cite{EilenbergNiven}. Completeness in general fails:
$ix-xi+1=0$ has no solution (an example from \cite{Widdows}) or have unexpected solution sets: 
$x^2=-1$ has an entire $2$-sphere as solutions and 
there are polynomials like $f(x)=x^2 i x i + i x^2 i x- i x i x^2-x i x^2 i$ for which $\mathbb{H}$ 
is the solution set (see \cite{Numbers}). 
By choosing a basis, $\mathbb{H}$ contains three linearly independent 
complex sub-rings generated by $1$ and a choice of spacial vectors $(0,b,c,d)$ 
satisfying $b^2+c^2+d^2=1$.

\paragraph{} 
It is natural to ask for a relation between the Frobenius and Hurwitz theorem.
Note that the Frobenius theorem does not assume ``normed" but makes an associativity assumption. 
But there seems no obvious link between the Frobenius and Hurwitz statements, as the class of
"normed division algebras" and "division algebras" is different and
associativity is quite a strong assumption. The question asks whether there
could be dependencies between the proofs of the theorems. The answer is probably "no":
going from Frobenius to Hurwitz requires to get rid of
the attributes "associative, finite dimensional" and adding "normed" instead.
Historically, Hurwitz paper was only published posthumously and does not 
cite Frobenius even so Hurwitz (1859-1919) and Frobenius (1849-1917) were contemporaries
and both worked in Zurich. Frobenius was at ETH Zuerich between 1875 and 1892. When
Frobenius took over Kronecker's chair in Berlin, Hurwitz in turn took over his chair at ETH
from 1892 until his death in 1919. 

\paragraph{}  
The lattice $\mathbb{I}$ of integers = {\bf integer quaternions} is a maximal order in $\mathbb{H}$,
where the notion of {\bf order} in non-commutative algebras has been defined by Emmy Noether: \cite{Lam2001}:
an order is just a sub-algebra which is a lattice; it is {\bf maximal} if it is not contained in a 
larger order. The Lipschitz integers in $\mathbb{H}$ form an order but only when adjoining the
Hurwitz quaternions, one gets a maximal order. 
The lattice is also known as the densest lattice packing of $\mathbb{H} = \mathbb{R}^4$ by unit spheres,
the integer quaternions are partitioned into two different type of integers, 
the {\bf Lipschitz quaternions} and the {\bf Hurwitz quaternions}. 
The Lipschitz quaternions consist of vectors $z=(a,b,c,d)$ with 
rational integers $(a,b,c,d)$. The Hurwitz quaternions are of
the form $(a,b,c,d)+(1,1,1,1)/2$ with rational integers $a,b,c,d$.
The set $\mathbb{P}$ of {\bf quaternions primes} are integer quaternions for which 
the {\bf arithmetic norm} $N(z)= a^2+b^2+c^2+d^2$ is prime. 
For $z=(a,b,c,d) \in \mathbb{P}$, the {\bf conjugate prime} is defined as $(a,-b,-c,-d)$. 
The symmetric group $V$ of permutations on $\{1,2,3,4\}$ acts on integer
quaternions and leaves primes invariant. Also the group $W$ generated by
elements in $V$ and the conjugation involution acts on $\mathbb{P}$. The group $W$ can not
be realized within $SU(3)$ but would need the larger group $U(3)$ and wipe out the charge
information. Having determinant $-1$, the conjugation involution $(a,b,c,d) \to (a,-b,-c,-d)$ is
not realized in $SU(3)$.

\paragraph{}  
Any positive definite quadratic form $N$ in $\mathbb{R}^n$ defines a lattice. In such a
lattice, the {\bf form-primes} are the integers
for which $N(z)$ is a rational prime. As Hurwitz showed in \cite{Hurwitz1922}, the primes in $\mathbb{H}$ 
are in one-to-one correspondence with the form-primes in $\mathbb{R}^4$. 
The quaternions are also unique among division algebras in that rational and 
arithmetic primes are the same. For Gaussian primes $\mathbb{C}$ for example,
real primes $p=4k+3$, primes of norm $2$ and primes of norm $p$ which is a prime of the form 
$p=4k+1$ are distinguished. Also the quadratic form $N(x,y) = x^2+xy+y^2$ on Eisenstein integers
generates only most of the Eisenstein primes, the split ones. In that case, $p=3$ are the 
ramified primes. The inert primes, the primes on the integer axes $\mathbb{Z}$ or $w \mathbb{Z}$ 
are not quadratic form primes among the Eisenstein integers, because there, the square root 
of $N$ is prime and not the norm $N$.

\paragraph{}  
A multiplication on $\mathbb{H}$ is defined by writing $z=a+bi+cj+dk$ with {\bf space units}
$i,j,k$ satisfying the Hamilton relations $i^2=j^2=z^2=ijk=-1$ and extending this linearly. 
The spacial part is generated by the units $i,j,k$. Together with $1$, they form a basis. 
While the quadratic equation $x^2=1$ has the only solution $\pm 1$ in $\mathbb{H}$,
the quadratic equation $x^2=-1$ in $\mathbb{H}$ already has already 
infinitely many solutions $(0,b,c,d)$ with $b^2+c^2+d^2=1$. 
The rational unit $1$ paired with one of the space units $(0,b,c,d)$ forms
a two-dimensional plane, which is also a complete sub-algebra isomorphic to $\mathbb{C}$. 
Except for $\mathbb{H}$ itself, these complex planes are the only complete normed division 
sub-algebras of $\mathbb{H}$. If we want the sub-algebra to have their integers as part of the
integers in $\mathbb{H}$, there are only three sub planes left and these are 
the ones generated by $1$ and a choice of $\{ i,j, k \}$. 

\paragraph{}  
If ``geometry" is a geometric space with an exterior derivative $d$ leading to Laplacian $L=(d+d^*)^2$, 
we have wave dynamics or isospectral Lax deformations of the derivative $d$. Examples of geometries are
compact Riemannian manifold or a finite simple graphs. 
If we look at the wave equation and want to be able to reach from a point $x$ in the geometry
to a point $y$, we need to use a complex variable. The wave equation $u_{tt}=-Lu$ with Laplacian $L=D^2$
has the explicit solution $u(t) = \cos(Dt) u(0) + i \sin(Dt) D^{-1} u'(0)$, where $D^{-1}$ is
the inverse of $D=d+d^*$ on the orthogonal complement of the kernel and $d$ is the exterior derivative. 
This can be written as a {\bf complex Schr\"odinger wave} $\psi(t) = e^{i D t} \psi(0)$ with
complex wave $\psi(t) = u(t) - i D^{-1} u'(t)$. Since complex sub-algebras exist in $\mathbb{H}$, we also 
look at Gaussian primes. Real sub-algebras like $\mathbb{R}$ generated by $1$ are 
not complete and some of their primes decay like $p=5$ which decays into $1 \pm 2i$. 
But primes which are initially in one of the three arithmetic complex wave subspaces, are
also primes in $\mathbb{H}$. 

\paragraph{} 
The linear map $\phi$, mapping $z \in \mathbb{H}$ to 
$\phi(z)=Z=(a/\sqrt{3}+ib,a/\sqrt{3}+ic,a/\sqrt{3}+id)) =(a_1+ib',a_2+ic',a_3+id') \in \mathbb{C}^3$
satisfies $|Z|^2=N(z)$. Conversely, given $Z=(a_1+ib',a_2+ic',a_3+id')$, define $a'$ by 
$a'^2 = a_1^2+a_2^2+a_3^2$, leading to  $\eta(Z)=(a',b',c',d')$ which satisfies $|Z|^2=N(\eta(Z))$. 
It satisfies $\eta(\phi(z))=z$. Given $U \in SU(3)$, define 
$\psi(U)=\eta U \phi: \mathbb{H} \to \mathbb{H}$. 
By construction, $\psi(U)$ preserves the unit sphere in $\mathbb{H}$ but it is neither linear 
nor invertible.  Any quaternion with prime amplitude $p$ can now be gauged to be an integer. 
The unitary group $U(3)$ allows to realize any permutation of $(a,b,c,d) \to \pi(a,b,c,d)$ of the four
coordinates of a quaternion as well as to perform sign changes like $(a,b,c,d) \to (a,-b,c,d)$. 
In other words, modulo gauge transformations given by the $U(3)$ action, every quaternion is 
equivalent to a quaternion $(a,b,c,d)$ with $0 \leq a \leq b \leq c \leq d$. 
The observation which led us to write this
down was that the $SU(2)$ equivalence classes of these classes come now in groups of $2$ or $3$
for odd primes which looks like the meson and baryon structure for Hadrons.

\section{Units and primes}

\paragraph{} 
We have seen that quaternions play a distinguished role among all algebraic structures: 
$\mathbb{C}$ and $\mathbb{H}$ are
the only algebraically complete associative division algebras, and as $\mathbb{H}$ 
is the maximal one containing $\mathbb{C}$ as proper complete sub algebras:
if $X$ is a sub-algebra of $\mathbb{H}$ containing $\{0,1\}$, then it must contain the real
line $\mathbb{R}$ and the square root of $-1$. In $\mathbb{H}$, any  solution $v$ of $v^2=-1$
x is of the form $(0,b,c,d)$ with $b^2+c^2+d^2=1$. 
If it contains besides $1$ and a solution to $v^2=-1$ a third linearly independent element, then 
by Frobenius, $X=\mathbb{H}$, as there is no other algebra between $\mathbb{C}$ and $\mathbb{H}$. 

\paragraph{} 
The algebra $\mathbb{H}$ has a natural representation in the algebra of complex $2 \times 2$ matrices:
every $z=(a,b,c,d) \in \mathbb{H}$ defines a matrix
$A=a \sigma_0 + b \sigma_1 + c \sigma_2 + d \sigma_3 \in GL(2,\mathbb{C})$, where $\sigma_0=I$ and 
$\sigma_i$ are the {\bf Pauli matrices}. Now $A(z) A(w) = A(z w)$ and $N(z) = {\rm det}(A(z))$. 
The {\bf Cauchy-Binet determinant formula} ${\rm det}(AB) = {\rm det}(A) {\rm det}(B)$ for
matrices $A(z)=\left[ \begin{array}{cc} a+ib & c+id \\ c-id & a-ib \end{array} \right]$
gives the identity $N(z w) = N(z) N(w)$. The unit sphere in $\mathbb{H}$ is $SU(2)$ which
happens to be a Lie group. The normed division algebras $\mathbb{C}$ and $\mathbb{H}$ 
can be distinguished by the fact that their unit spheres are continuous Lie groups
as this property eliminates both $\mathbb{R}$ and $\mathbb{O}$. The only Lie groups which 
are spheres are $S^k$ for $k=0,1,3$ and the only continuous ones are $S^k$ with $k=1,3$. 
(A short proof of that statement was sketched by \cite{DeVito2010}: if the Lie group $G$ is Abelian 
then the Lie algebra must be $\mathbb{R}^n$ so that it must be a universal cover of $G$. This only works 
if $G$ is not simply connected, implying $G=S^1$ as this is the only not simply connected but 
connected sphere. If $G$ is non-Abelian, define the 3-form $t(x,y,z)=([x,y],z)$ which by the 
non-Abelian assumption is not the zero form.  One can show that it is left and right invariant and 
so closed but not exact, leading to a nontrivial cohomology $H^3(G)$. 
This forces $G=S^3$ as this is the only Euclidean sphere with non-vanishing $H^3$.) 

\paragraph{}  
There are 24 units in $\mathbb{H}$; 16 of them are {\bf Hurwitz units} $(a,b,c,d)/2$ 
with $a,b,c,d \in \{-1,1\}$ and 8 of them are  {\bf Lipschitz units} $\pm 1, \pm i, \pm j \pm k$.
The Lipschitz units are permutations of $(\pm 1,0,0,0)$, the
Hurwitz integers are permutations of $(\pm 1,\pm 1, \pm 1, \pm 1)/2$.
The finite set $U$ of units forms the {\bf binary tetrahedral group} which
can also be written as the special linear group $SL(2,3)$ over $\mathbb{Z}/(3\mathbb{Z})$. 
As quaternions can be represented as $2 \times 2$ matrices,
$U$ is also a discrete finite subgroup of the unitary group $SU(2)$.  

\paragraph{} 
The group $SU(2)$ is also known as the compact symplectic group $Sp(1)$ or the 
spin group ${\rm Spin}(3)$. 
The finite group $U$ of units in $\mathbb{H}$
can be identified as the semi-direct product of the {\bf quaternion sub group}
$Q$ built by the $8$ Lipschitz units and the cyclic group $Z_3$, generated by conjugation
$i \to j \to k \to i$. It is finitely presented satisfying
$(ab)^2=a^3=b^3=1$ for the generators $a=(1+i+j+k)/2$ and $b=(1+i+j-k)/2$. 
All cyclic subgroups have order $2,3$ or $6$. The element $-a$ for example generates 
a cyclic group of order $6$. The only units for which all entries are non-negative are
$(1,0,0,0),(0,1,0,0),(0,0,1,0),(1,1,1,1)/2$.  The fact that each of the 16 regions like
$a>0,b>0,c>0,d>0$ contains exactly one Hurwitz unit will play a role combinatorially
when looking at weak equivalence equations. 

\paragraph{}   
In $\mathbb{C}$, the integers have only one type, but primes
appear in three flavors: {\bf inert}, {\bf split} or {\bf ramified} depending on how they 
factor in the field extension: the real ones with prime $\sqrt{N(p)}$ of the form $4k+3$, the
ones with prime $N(p)$ of the form $4k+1$ and then the primes of norm $N(z)=2$. There are
4 units $\{1,i,-1,-i \;\}$. The group $U$ of units is the cyclic group $C_4$. There is an
other group $W$, the group of all permutations of the coordinates. 
In $\mathbb{C}$, it is the dihedral group $D_4$. 
Every equivalence class of $W$ can be represented as $(a,b)$ with $0 \leq a \leq b$. 
All $W$ equivalence classes have only one element. The $U$ equivalence classes have $1$
or $2$ elements. The units and the primes with norm $2$ have only one element, also the
real primes have only one equivalence class. The other primes have $2$ elements in each
equivalence class. An example is $4+i$ and $4-i$. 

\paragraph{}  
Hurwitz showed that in $\mathbb{H}$ and $p \neq 2$, there are exactly $24 (p+1)$ primes 
$(a,b,c,d)$ with $a^2+b^2+c^2+d^2=p$. There are therefore $p+1$ classes of primes above
each rational prime $p$. For example, for $p=3$  
we have the $3+1$ primes $(1,1,1,0),(1,1,0,1),(1,0,1,1),(0,1,1,1)$. 
By multiplying with units, we see that they are all not equivalent. 

\section{Prime factorization}

\paragraph{}  
While in $\mathbb{C}$, all factorizations are equivalent, the structure is 
more interesting in the non-commutative quaternion case, where
Conway and Smith finalized in \cite{ConwaySmith} the factorization structure 
of $\mathbb{I}$: it consists of unit migration, recombination and meta commutation.
Recombination means to replace a pair $z \overline{z}$ in the prime factorization
with an equivalent pair $w \overline{w}$ and move it to an other place. 
Lets look next a unit migration: 

\paragraph{} 
Given an integer quaternion $z$, the set of left conjugacy classes $Uz$ is different 
from the set of right equivalence classes $zU$. This means that if $z$ is an integer
and $u$ is a unit, then $u z$ can not be written as $z v$ for a unit $v$ in general. 
What happens is that that the entries get scrambled around. This can be realized using
permutations, respectively using elements of the group action of $U(3)$ on $\mathbb{H}$. 

\begin{lemma}[Unit migration]
For any unit $u$ in $\mathbb{H}$, there exists a permutation $\pi$ of the four 
coordinates $z=(a,b,c,d)$ and an other unit $v$ such that for all $z \in \mathbb{I}$ 
the identity $u z = \pi( z v)$ holds. 
\end{lemma}

\begin{proof}
Since there are finitely many units, this can be checked case by
case $u z = \pi( v z )$. Here is the explicit list:  \\
\begin{center} 
\begin{tiny}
\begin{tabular}{lll}
unit u             &   unit v           &   permutation $\pi$  \\ \hline
$\{-1,0,0,0\}$       &  $\{-1,0,0,0\}$      &  $\{1,2,3,4\}$   \\
$\{-1,-1,-1,-1\}/2$  &  $\{-1,-1,1,1\}/2$   &  $\{1,4,2,3\}$   \\
$\{-1,-1,-1,1\}/2$   &  $\{0,0,0,-1\}$      &  $\{2,4,3,1\}$   \\
$\{-1,-1,1,-1\}/2$   &  $\{1,1,-1,1\}/2$    &  $\{4,2,1,3\}$   \\
$\{-1,-1,1,1\}/2$    &  $\{-1,1,1,1\}/2$    &  $\{2,3,1,4\}$   \\
$\{-1,1,-1,-1\}/2$   &  $\{0,0,0,1\}$       &  $\{3,1,2,4\}$   \\
$\{-1,1,-1,1\}/2$    &  $\{1,-1,-1,-1\}/2$  &  $\{3,2,4,1\}$   \\
$\{-1,1,1,-1\}/2$    &  $\{1,1,-1,-1\}/2$   &  $\{4,1,3,2\}$   \\
$\{-1,1,1,1\}/2$     &  $\{-1,-1,1,-1\}/2$  &  $\{1,3,4,2\}$   \\
$\{0,-1,0,0\}$       &  $\{1,0,0,0\}$       &  $\{4,3,2,1\}$   \\
$\{0,0,-1,0\}$       &  $\{-1,1,1,-1\}/2$   &  $\{2,1,4,3\}$   \\
$\{0,0,0,-1\}$       &  $\{1,-1,-1,1\}/2$   &  $\{3,4,1,2\}$   \\
$\{0,0,0,1\}$        &  $\{1,-1,-1,1\}/2$   &  $\{3,4,1,2\}$   \\
$\{0,0,1,0\}$        &  $\{-1,1,1,-1\}/2$   &  $\{2,1,4,3\}$   \\
$\{0,1,0,0\}$        &  $\{1,0,0,0\}$       &  $\{4,3,2,1\}$   \\
$\{1,-1,-1,-1\}/2$   &  $\{-1,-1,1,-1\}/2$  &  $\{1,3,4,2\}$   \\
$\{1,-1,-1,1\}/2$    &  $\{1,1,-1,-1\}/2$   &  $\{4,1,3,2\}$   \\
$\{1,-1,1,-1\}/2$    &  $\{1,-1,-1,-1\}/2$  &  $\{3,2,4,1\}$   \\
$\{1,-1,1,1\}/2$     &  $\{0,0,0,1\}$       &  $\{3,1,2,4\}$   \\
$\{1,1,-1,-1\}/2$    &  $\{-1,1,1,1\}/2$    &  $\{2,3,1,4\}$   \\
$\{1,1,-1,1\}/2$     &  $\{1,1,-1,1\}/2$    &  $\{4,2,1,3\}$   \\
$\{1,1,1,-1\}/2$     &  $\{0,0,0,-1\}$      &  $\{2,4,3,1\}$   \\
$\{1,1,1,1\}/2$      &  $\{-1,-1,1,1\}/2$   &  $\{1,4,2,3\}$   \\
$\{1,0,0,0\}$        &  $\{-1,0,0,0\}$      &  $\{1,2,3,4\}$   \\
\end{tabular}
\end{tiny}
\end{center} 
\end{proof}

\paragraph{}  
There are prime pairs which can be disappear and reappear anywhere as they temporarily
become rational integers:

\begin{lemma}[Recombination]
A pair $z \overline{z}$ in the prime factorization can be moved anywhere and 
replaced with an other pair $w \overline{w}$ of the same norm. 
\end{lemma} 
\begin{proof}
A rational number commutes with any number and $z \overline{z}$ is a rational prime
if $z$ is a prime. 
\end{proof} 

\paragraph{} 
As Hurwitz already noticed, primes of norm $2$ play a special role. As in the
complex $\mathbb{C}$ case, all primes over the prime $p=2$ are equivalent. They
especially do not have any charge.
Commuting with them does not change much of the quaternion. It only produces a
rotation by 90 degrees in a two dimensional spacial sub plane (3 cases) or a
rotation by 180 degree reflection. The situation is otherwise similar to the unit 
migration:

\begin{lemma}[Prime 2 migration]
For any prime $p$ of norm $2$, there is an other prime $q$ of norm $2$ and a
rotation $\pi$ in a two-dimensional sub plane such that $p z = \pi (z q)$ for all $z \in \mathbb{I}$.
\end{lemma}

\begin{proof}
Here are the 6 cases \\
\begin{center}
$(1,1,0,0)(a,b,c,d) = (a,b,-d,c)(1,1,0,0)$\\
$(1,0,1,0)(a,b,c,d) = (a,d,c,-b)(1,0,1,0)$\\
$(1,0,0,1)(a,b,c,d) = (a,-c,b,d)(1,0,0,1)$\\
$(0,1,1,0)(a,b,c,d) = (a,c,b,-d)(0,1,1,0)$\\
$(0,1,0,1)(a,b,c,d) = (a,d,-c,b)(0,1,0,1)$\\
$(0,0,1,1)(a,b,c,d) = (a,-b,d,c)(0,0,1,1)$\\
\end{center}
\end{proof}

\paragraph{} 
The fundamental theorem of arithmetic in quaternions requires to 
understand what happens if two odd primes are switched. This process is called
meta commutation:  

\begin{lemma}[Meta commutation]
For any two primes $z,w$ above $p,q$ there is a permutation $\pi$ of the $p+1$ prime classes
above $p$ such that $p q = q \pi(p)$ modulo $U$. 
\end{lemma}
\begin{proof}
The integers of norm $1$ and $2$ have already been dealt with. 
The permutation structure has been studied in \cite{CohnKumar}. 
\end{proof}

If we use two $U$-equivalent primes $p,w$ in the factorization, then the Moebius 
function $\mu(z)$ is zero.  \\

\paragraph{} 
Here is {\bf fundamental theorem of arithmetic} for the complete associative
maximal normed division algebra: 

\begin{thm}[Conway-Smith]
Up to unit migration, recombination and meta-commutation, the factorization
of a quaternion integer into primes is unique. 
\end{thm} 
\begin{proof}
If the order of the primes is given, we have uniqueness up to unit migration. 
Meta commutation takes care of the rest. See \cite{ConwaySmith}. 
\end{proof} 

\paragraph{} 
How do we represent the $U(3)$ equivalence classes of a prime modulo the group $U$?

\begin{lemma}
Every prime $Q$ in $\mathbb{H}$ is $U$-equivalent to one or two or
three $U(3)$ representatives in the positive quadrant $\mathbb{I}^+$.
\end{lemma}
\begin{proof}
Let $Q$ be the region in $\mathbb{H}$, where all coordinates $a,b,c,d$ are nonnegative.
Given a prime $z=(a,b,c,d) \in Q$, we can write down all 
the 24 left conjugates $u z$ and 24 right conjugates $z v$, where $u,v$ run over the 24 units. 
Multiplying with Lipschitz units different from $\pm 1$ gives 
$u z \in \{ (b,-a,d,-c),(c,-d,-a,b),(d,c,-b,-a),(-d,-c,b,a),(-c,d,a,-b),(-b,a,-d,c) \}$ and 
$z u \in \{ (b,-a,-d,c),(c,d,-a,-b),(d,-c,b,-a),(-d,c,-b,a),(-c,-d,a,b),(-b,a,d,-c) \}$ which are not in $Q$
unless two coordinates are zero. But then the conjugate is in the same $U(3)$ representative.
For the $16$ Hurwitz units $u$, the list of $2u z$ is \\
$(-a + b + c + d,-a - b + c - d,-a - b - c + d,-a + b - c - d)$ \\
$(-a + b + c - d,-a - b - c - d,-a + b - c + d,a + b - c - d)$ \\
$(-a + b - c + d,-a - b + c + d,a - b - c + d,-a - b - c - d)$ \\
$(-a + b - c - d,-a - b - c + d,a + b - c + d,a - b - c - d)$ \\
$(-a - b + c + d,a - b + c - d,-a - b - c - d,-a + b + c - d)$ \\
$(-a - b + c - d,a - b - c - d,-a + b - c - d,a + b + c - d)$ \\
$(-a - b - c + d,a - b + c + d,a - b - c - d,-a - b + c - d)$ \\
$(-a - b - c - d,a - b - c + d,a + b - c - d,a - b + c - d)$ \\
$(a + b + c + d,-a + b + c - d,-a - b + c + d,-a + b - c + d)$ \\
$(a + b + c - d,-a + b - c - d,-a + b + c + d,a + b - c + d)$ \\
$(a + b - c + d,-a + b + c + d,a - b + c + d,-a - b - c + d)$ \\
$(a + b - c - d,-a + b - c + d,a + b + c + d,a - b - c + d)$ \\
$(a - b + c + d,a + b + c - d,-a - b + c - d,-a + b + c + d)$ \\
$(a - b + c - d,a + b - c - d,-a + b + c - d,a + b + c + d)$ \\
$(a - b - c + d,a + b + c + d,a - b + c - d,-a - b + c + d)$ \\
$(a - b - c - d,a + b - c + d,a + b + c - d,a - b + c + d)$. \\
In order that we are in $Q$, we need to satisfy in each case 4 inequalities
simultaneously. Its possible for example in the last entry if $a>b+c+d$ or
in the entry $(a + b + c - d,-a + b - c - d,-a + b + c + d,a + b - c + d)$ if
$b>a+c+d$. The list of $2z v$ is \\
$(-a + b + c + d,-a - b - c + d,-a + b - c - d,-a - b + c - d)$ \\
$(-a + b + c - d,-a - b + c + d,-a - b - c - d,a - b + c - d)$ \\
$(-a + b - c + d,-a - b - c - d,a + b - c - d,-a + b + c - d)$ \\
$(-a + b - c - d,-a - b + c - d,a - b - c - d,a + b + c - d)$ \\
$(-a - b + c + d,a - b - c + d,-a + b - c + d,-a - b - c - d)$ \\
$(-a - b + c - d,a - b + c + d,-a - b - c + d,a - b - c - d)$ \\
$(-a - b - c + d,a - b - c - d,a + b - c + d,-a + b - c - d)$ \\
$(-a - b - c - d,a - b + c - d,a - b - c + d,a + b - c - d)$ \\
$(a + b + c + d,-a + b - c + d,-a + b + c - d,-a - b + c + d)$ \\
$(a + b + c - d,-a + b + c + d,-a - b + c - d,a - b + c + d)$ \\
$(a + b - c + d,-a + b - c - d,a + b + c - d,-a + b + c + d)$ \\
$(a + b - c - d,-a + b + c - d,a - b + c - d,a + b + c + d)$ \\
$(a - b + c + d,a + b - c + d,-a + b + c + d,-a - b - c + d)$ \\
$(a - b + c - d,a + b + c + d,-a - b + c + d,a - b - c + d)$ \\
$(a - b - c + d,a + b - c - d,a + b + c + d,-a + b - c + d)$ \\
$(a - b - c - d,a + b + c - d,a - b + c + d,a + b - c + d)$. There are again conditions
for which this works like in the last case, if $a$ dominates. Depending on whether we
are in a class already covered in $2uz$, we have now equivalence classes
with 2 entries $\{z,uz\}$ or equivalence classes with three entries $\{z,uz,zv\}$. 
There are no other cases as from the list of 256 entries $u z v$ (a list we omit),
we are either equivalent to $z$ or not in $Q$. 
\end{proof}

In the meson case, one of the entries is always zero. 
If two entries are zero like $(a,0,c,0)$, the partner must be conjugated by a Hurwitz
unit and is then of the form $(a',a',c',c')$. 

\section{Charge functional}

\paragraph{} 
On $\mathbb{Z}$ where the set of rational primes $\{2,3,5 \dots \}$ are the fundamental building
blocks, there is a natural multiplicative charge, the M\"obius function $\mu(x)$.
For complex primes, we can define a charge by looking at the sign of the
imaginary part when the argument is gauged to be in $(-\pi/4,\pi/4)$. The prime $2+i$ for example
has charge $1$ while the prime $2-i$ has charge $-1$.  The prime $1+i$ has zero charge as $1+i$
is conjugated to $1-i$. For integers in $\mathbb{C}$ it an be enhanced by defining $c(x)$ as 
$\mu(N(x)) (-1)^{n(x)}$, where $N(x)$ is the arithmetic norm and 
$n(x)$ is the number of negatively charged primes in 
the factorization. Also this is a multiplicative function.

\paragraph{} 
We can also define a charge function on the set $\mathbb{I}^*$ of non-zero quaternion integers
which are square free. We only take square free integers because quarks are fermions. Here is a
preliminary definition: lets assume that {\bf charge} is a function $c: \mathbb{I}^* \to \mathbb{R}$ satisfying:
   a) $c(z w) = c(z) +c(w)$ for non-units,
   b) $c(z) = -c(\overline{z})$,
   c) $c(z) = c(z')$ if $z' = \pi(z)$ with $\pi \in V$,
   d) primes which are $W$-equivalent but not $V$-equivalent have opposite charge,
   e) If a Lipschitz prime is $U$-equivalent to a Hurwitz prime, the absolute charge of the 
      Lipschitz prime is not smaller than the Hurwitz prime,
   f) The sum of the charges of elements in a $U$-equivalence class is an integer,
   g) $c$ is positive on the positive quadrant, and 
   h) $c$ is the smallest function of this kind. \\

\paragraph{} 
A charge for units could be defined by $c(w)-c(z)$ if $w=uz$, but this needs independence of $w,z$. 
A direct definition of charge for units does not work as $u^3=1$ and $c(1)=0$ 
would imply $c(u)=0$. For example, the unit $a=(1,1,1,1)/2$ has the property $a^3=-1$ and 
$-1$ has charge $0$. If the multiplicative nature would be extended to this, 
the unit would have charge $0$.

\paragraph{} 
There is only one charge function:
existence is shown by construction. Because it takes the value $1$ for some 
equivalence classes, there can not be any smaller one. 
To get uniqueness, we can restrict to primes by property a). 
For the equivalence classes with one element, we have either units or prime $2$.
Since conjugation leaves a prime $2$ invariant, this must have charge $0$.
Lets look at the equivalence classes with one Lipschitz and one Hurwitz prime. 
If the two have opposite charge, then we get a $2/3,-2/3$ or $1/3,-1/3$ 
pair. If not, we get a $2/3,1/3$ pair. 
In an equivalence class with one Lipschitz and two Hurwitz primes, we have
either the Lipschitz one equivalent to one of the Hurwitz, which gives $2/3,2/3,-1/3$
or then the two Hurwitz equivalent, which gives $2/3,-1/3,-1/3$. 

\paragraph{}  
Now we an extend the charge function $C$ from primes to equivalence classes of primes. 
This can be done by adding up the charges of the elements of the equivalence class
$C(z) = \sum_{w \sim z} c(w)$.
This is motivated as follows. We look at a factorization into equivalence classes. 
Then we sum over all possible paths of factorizations,
we get for each Hadron an integer charge, even so the charge of 
an individual quark is a fraction, a probability of hitting that particular quark. 
By definition, this is now an integer-valued function on equivalence classes of primes. \\

\paragraph{} 
An other natural functional, additive with respect to the multiplication is the logarithm
of the norm. We could abstractly define {\bf mass} as a function $m: \mathbb{I} \to \mathbb{R}$ satisfying:
   a) $m(z w) = m(z) + m(w)$,
   b) $m(z) \geq 0$,
   c) $m(z)=0$ if and only if $z$ is a unit,
   d) $e^{m(z)}$ takes integer values on integers,
   e) $m$ is the smallest function of this kind. 
and show that there is only one mass function on $\mathbb{I}^*$, the function $m(z)=\log(N(z))$. Whether
``mass" is a good name for such a functional will become only clear when looking at solutions of wave
equations. 

\section{Leptons}

\paragraph{}  
We think of a square-free Gaussian integer as a collection of {\bf leptons}, where the individual
Gaussian prime equivalence classes are indecomposable {\bf Fermions}. A prime $z$ of type $4k+1$ together with
an opposite charged particle $\overline{z}$ forms an {\bf electron-positron} pair. We gauge the integers
with units $U=\{1,i,-1,-i\}$ to be in the sector $\pi/2<{\rm arg}(z) \leq \pi/2$. The {\bf charge}
of a lepton $z$ is then defined as the sign of the argument of $z$ in the branch $(-\pi,\pi)$.
We think of the logarithm of the norm $N(z)$ or $|z|$ as {\bf mass}, whatever is prime. 
A rational positive prime $4k+3$ plays the role of a {\bf neutrino}.
It is lighter because its momentum $|z|$ is prime while for $4k+1$ primes which are electrons or
positrons, the energy $N(z)$ is prime. A neutrino is {\bf neutral} as it is located 
on the real axes. The largest known prime for example is a Mersenne prime and so a neutrino.

\paragraph{}
An integer $n=p_1 \dots p_k$ is a {\bf lepton configuration}. The fact that the Gaussian primes 
form a {\bf unique factorization domain} translates into the statement that any lepton configuration 
can be decomposed uniquely into such leptons as well as a bunch of neutral mystery particle $2$ which is its
own anti-particle. The uniqueness holds only modulo {\bf gauge transformations} which act here as
multiplications by units. We will see that in the Quaternion case, this fact is no more the case, 
because that is, where the Hadrons will come in explaining why quarks form baryons and mesons.
The electron-positron pair is {\bf not bound} together: there is no unit which maps one into the other.
Factoring out the symmetry of units renders the factorization unique. The product $(-3) (-7)$ for
example is gauge equivalent to the product $3 \cdot 7$. Let us now move from primes to rational integers and
call a {\bf rational integer} $n \in \mathbb{Z}$ a {\bf Boson configuration} if it contains an even
number of Fermionic prime factors counted with multiplicity, otherwise it is a Fermion. 

\paragraph{}
Mathematically one could declare a natural number 
$n \in \mathbb{N}$ to be a Fermion if its {\bf Jacobi symbol} 
$(-1|n)=\left( \frac{-1}{n} \right)$ is $-1$.
Otherwise, if $(-1|p)=1$, it is a Boson. The Gauss law of {\bf quadratic reciprocity} 
tells now that two odd primes $p,q$ satisfy the {\bf commutation relations}
$(p|q) =(q|p)$ if at least one of them is a Boson and that the
{\bf anti-commutation relation} $(p|q)=-(q|p)$ hold exactly
if both primes $p,q$ are Fermions. In other words, if we look at the 
Jacobi symbol as an operator $p \cdot q$, then Bosons commute with everything else, but the sign changes,
if we switch two Fermions.

\paragraph{}
The {\bf two square theorem of Fermat} telling that an integer $n$ can be represented as $a^2+ b^2$
if and only if $n$ is a Bosonic integer can be interpreted as the fact that a Bosonic rational prime
is actually composed of two leptons $a+ib,a-ib$, where $a^2+b^2=p$. The positron and electron are anti-particles
of each other, but they are not equivalent, since there is no gauge from one to the other. If we factor out the
{\bf gauge symmetries} given by the {\bf units}, then the factorization aka particle decomposition of the
lepton set is unique. This is the
{\bf fundamental theorem of arithmetic} for Gaussian integers. It can be proven from the rational
case using the $1-1$ identification of $\mathbb{C}/D_4$ with the set of rational primes. 

\paragraph{}
The {\bf Pauli exclusion principle} is 
encoded in the form of the {\bf Moebius function} $\mu_G(n)$ which is equal to $1$ if a Gaussian integer 
$n$ is the product of an even number of different Gaussian primes, and $-1$ if it is the product 
of an odd number of different Gaussian primes and $0$, if it contains two identical particles. 
Again this {\bf particle allegory} is already useful as a {\bf mnemonic} to
remember theorems like the two square theorem, or the quadratic reciprocity theorem:
"Quadratic Reciprocity means that only Fermion primes anti-commute $(p|q)=-(q|p)$.
Fermat's two square theorem assures that Bosonic rational primes $p=a^2+b^2$
are composed of two Gaussian primes $a \pm i b$. The others are all real, light and neutral." 

\section{Hadrons} 

\paragraph{}
Hadrons are quaternion prime equivalence classes $z$ with norm $N(z)$ different from $2$. The prime $2$ is special 
also in quaternions: as Hurwitz already showed one can despite non-commutativity place the factors $2$
outside: every integer quaternion $z$ is of the form $z=(1+i)^r b$ where $w$ is an odd 
integer quaternion. There are two symmetry groups acting on hadrons. 
One is the group $U$ of units, the other is the group $V$ generated by coordinate permutations and conjugation. 
The group $U$ has 24, the group $V$ has $48$ elements. The groups are no more
contained in each other like in the complex case, where $U=\mathbb{Z}_4$ was a subgroup of $V=D_4$.

\paragraph{}
If we look at the $U$-equivalence classes first and then look at the orbits of $V$, we see that some
particles are fixed under $V$ or then that they come in pairs. We will interpret this as {\bf particle
anti-particle pairs}. We can however also look at the $V$ equivalence classes first and then 
look at the orbits of $U$, then we see 1 or 2 or 3 particles combined. This is remarkable. Lets repeat
the statement: the $SU(3)$ equivalence classes of primes of positive charge are given by elements for 
which all coordinates are non-negative. If we call two such particles $z,w \in SU(2)$ equivalent if there exists
a unit in $\mathbb{H}$ such that $z = u w$, then there are three type of equivalence classes: if $N(z)=2$,
then there is one equivalence class. If $z$ is a unramified prime, then there are two cases: either there
are two elements in each $U$ equivalence class or three elements in each $U$-equivalence class. 
So, it is important first to look at the ``strong equivalence classes", then classify them according to 
``weak equivalence". If we would look at ``weak equivalence classes" first, then classify according to 
``strong equivalence", then we would get in the unramified case only two types. They tell the sign of the charge. 

\paragraph{} 
A {\bf hadron} is a quaternion prime equivalence class which contains more than one prime.
As a matter of fact, each hadron to an odd prime $p$ consists either of two or three 
{\bf quarks}. Quarks can be either Lipschitz or Hurwitz primes. 
The {\bf Lagrange four square theorem} assures that there are no neutrini type particles among Hadrons:
particles for which three coordinates are zero. 
So, we can think of odd quaternion primes as {\bf quarks} which form equivalence classes in the form of 
baryons or mesons. The conjugate quaternion is the {\bf anti particle}. 
This is a  probabilistic picture: take an integer $n \in \mathbb{I} \subset \mathbb{H}$. 
Look at all possible prime factorizations of $n$. We can look at such a factorization also 
modulo units meaning to look at equivalence classes. 
If we fix the order of the equivalence classes, then we see an ordered sequence
of mesons (equivalence classes with two primes), baryons (equivalence classes with three primes)
and bosons related to $2$ (equivalence classes with one prime). The latest primes can be moved anywhere without
producing other changes  as switching it with an other prime only produces a slight gauge transformation. 
Switching the order of two other classes however can change the nature of the particles. 

\paragraph{}
As can be seen from 
the Hurwitz factorization theorem which shows that in such a factorization, one can avoid have
factors $z \overline{z}$ near each other. The allegory is that they would ``annihilate" into energy
and re-emerge as particles. 
Going from one factorization to an other is a rather complex interaction process
changing the nature of some baryons involving gauge bosons. 
It does certainly not match with all particle processes as the number of primes in a factorization 
of $z$ is constant among all factorizations. This could just be the limit of particle phenomenology
alone. The true nature of physics only comes with dynamics. 

\paragraph{}
Baryons are Fermions and mesons are Bosons. Like in the complex case, we have 
a mystery $p=2$ case, which has only one equivalence class. Lets call it the {\bf $2$ particle}
 even so we would like to associate it with something real. Closest to mind comes ``Higgs" which 
is neutral, light, its own anti-particle and can gives more mass to other particles.
Here it does so by multiplying with an other integer. 

\paragraph{}
The {\bf group of units} contains $8$ particles of the form 
$(\pm 1,0,0,0)$, $(0,\pm 1,0,0)$, $(0,0,\pm 1,0)$,$(0,0,0,\pm 1)$ which modulo $U$ 
are all equivalent to the neutral $(1,0,0,0)$, the $Z$-Boson. 
There are $16$ remaining units. Modulo $V$ they are all equivalent to $(1,1,1,1)/2$. This has
a positive charge and is the {\bf $W^+$ boson}. Its conjugate is the {\bf $W^-$ boson}.

\paragraph{}
Now lets look at a meson $(pq)$ containing a Lipschitz prime $p$ and Hurwitz prime $q$. 
Since $p$ and $q$ are equivalent in $U$, there exists a permutation, possibly with a
conjugation, such that $\overline{p}$ is gauge equivalent to $q$. If a conjugation is 
involved, then $p$ and $q$ have the same sign of charge, otherwise opposite.
Lets postulate that the {\bf charge} of a Lipschitz quark is $2/3$. 
Since we have identified modulo $V$ we can assume that it is positive. 
In the meson case, the charge of the other particle is $1/3$ if it has the same charge 
and $-2/3$ if it has opposite charge. 

\paragraph{}
In the baryon case, if the two other quarks have the same charge sign, they have charge $-1/3$.
If one has the same charge than the Lipschitz one, then both have charge $2/3$ and the other $-1/3$. 
The charge of an equivalence class is the sum of the charges. We have now assigned a charge in a gauge 
invariant way: a Lipschitz quark $(a,b,c,d)$ has charge $+2/3$ if $a \leq b \leq c \leq d$ and $-2/3$
if it is obtained from that by switching two coordinates. The structure of the equivalence classes 
assures a compatible choice so that the total charge is an integer. 

\paragraph{}
In the meson case, we observe that one of the Lipschitz primes is 
located on the three or two dimensional coordinate plane. If a coordinate is zero, then we can
not perform all flipping operations and we have to see whether we have to flip it in the {\bf Gaussian 
sub-plane}. In this picture, all Hadrons have charge $0,-1,1$. There are no Hadrons of charge $2$. 
Particle physicists mention {\bf beta $uuu$-hadrons} detected in experiments. 
Also exotic 4 and 5 quark matter consisting of more than 3 quarks has been detected. 
It is likely just an analogue of 6 quark
state which is known under the name ``Deuterium". But the quarks are the not  
``in one bag". Its more like that a neutron and proton glued together by strong forces. 

\paragraph{}
The non-uniqueness of prime factorization allows to see the transition from one to an 
other factorization as a {\bf particle process}. It involves the gauge bosons. 
It only becomes only unique modulo {\bf unit migration} and recombination.
(see \cite{ConwaySmith}). This means that if $x$ is a quaternion integer and $N(x)=p_1 \dots p_n$
then $x= (P_1 u_1) (u_1^{-1} P_2 u_2)  \dots  (u_{n-1}^{-1} P_n)$, where the $u_j$ are units
and $P_i$ are Hurwitz primes. In other words, the factorization becomes unique if
we look at it on the meson/baryon scale but it depends on the order. It becomes unique
when including meta-commutation: the prime factorization of a nonzero Hurwitz integer
is unique up to meta-commutation, unit migration and recombination, the process
of replacing $P \overline{P}$ with $Q \overline{Q}$ if $P,Q$ have the same norm.  \\

\paragraph{}
The $p=2$ is special $(1,1,0,0)$ and not included in the above remarks. It is neutral
and equivalent to its anti-particle. It is "Higgs like", as in the lepton case and
we can not place it yet. Also, we like to think about a situation, where space is actually 
a dyadic group of integers (a much more natural space as it is compact, features a smallest
translation), here scaling by a factor $2$ is a symmetry and multiplication by $2$ 
makes the grid finer (making the $2$-adic norm small). The mechanism of mass can anyway
only be understood when looking at dynamical setups, where particles travel. When looking
at wave equations in dyadic groups, the multiplication by 2 plays a special role and it can 
slow down particles, similar as mass does. \\

\paragraph{}
Here are some examples.
For $p=3$, $(0,1,1,1)$, $(1,1,1,3)/2$ form a meson of charge $1=2/3+1/3$.
For $p=5$, we have a meson $(0,0,1,2),(1,1,3,3)/2$ of charge $1$. 
For $p=13$, we have a baryon $(1,1,1,7)/2$, $(1,2,2,2)$, $(3,3,3,5)/2$ of charge $0$ 
and a meson $(0,0,2,3)$, $(1,1,5,5)/2$ of charge $1$.
For $p=41$ we have a meson $(0,3,4,4)$, $(3,3,5,11)/2$ of charge $\pm 1$. 
and a meson $(0,1,2,6)$, $(3,5,7,9)/2)$ for which all coordinates are different. 

\paragraph{}
The integer units $i,j,k$ are the $Z$ vector bosons while the $(1+i+j+k)/2$ etc are $W^{\pm}$ 
vector bosons. In the complex the $1,i$ generate photons. 
We can not have Lipschitz primes $(a,a,b,b)$ as this is $2a^2+2b^2 = 2(a^2+b^2)$ which is 
$0,2$ modulo $4$.  You see in the figures some pictures. We are able to attach to any 
prime a collection of baryons, for which the charges is determined.

\paragraph{}
A square free integer is a collection of Fermi particles. A given factorization is a particular particle
configuration. Going from one factorization to an other is an interaction process which 
invokes exchange of vector bosons and where particles change nature. 
Say, if $z$ is a baryon and $w$ a meson, then $z'$ and $w'$ can represent a meson-baryon
pair.

\paragraph{}
Gluons come from more general symmetry transformations. They do not have 
mass similarly like photons do not have mass. Lets postulate that Gluons are obtained
from symmetry transformations when $SU(3)$ acts as a symmetry in
$\mathbb{H}$. This is motivated from the complex place where $(a,b) \to (-b,a)$, the
multiplication with i, the square root of -1, is realized as an element in the dihedral
group, actually the cyclic group U. We can associate a gluon with a transformation 
$(a,b,c,d) - (a',b',c',d')$ which is a subgroup of $SU(3)$ The vector bosons $W^0,W^{pm}$
are implemented by multiplications by a unit in $SU(2)$. 

\section{Remarks and questions}

\paragraph{Related ideas}
In 1952, Paul Kustaanheimo \cite{Kustaanheimo} explored the idea to do physics over
a finite field $K$ rather than over the real or complex numbers.
He mused about the size of the size of the prime. The situation is different however in that 
the approach is to take the integers in $\mathbb{H}$ are the central point. 
We also should mention that particle phenomenology in division algebras has been worked on 
before, like in \cite{DixonDivisionalgebras}. The octonions could play an important role
too but their combinatorics is more complicated as there are three type of primes: 
{\bf Gravesian integers} $(a,b,c,d,e,f,g)$ which play the role of the
Lipschitz primes in the Hurwitz case, the
{\bf Kleinian integers} $(a,b,c,d,e,f,g) + (1,1,1,1,1,1,1,1)/2$ which play the role of the 
Hurwitz primes in $\mathbb{H}$. Then there are the 
{\bf Kirmse integers} which includes elements for which $4$ of the entries are half integers.
Eight {\bf maximal orders} have been identified by Kirmse \cite{Kirmse25}, a miscount 
which was later corrected by Coxeter to seven \cite{Coxeter46}). \\

\paragraph{Quaternion valued waves}
We have seen that any $\mathbb{C}$-valued quantum evolution extends to an
evolution of quaternion valued fields when seeing the evolution on $\mathbb{C}^3$ 
three complex planes but then merge things in a common real line. 
It would be useful to explore this more systematically. 
The physics in Division algebras is any new \cite{Baez2012,DixonDivisionalgebras}. 
One certainly also has to consider the possibility of octonion valued fields. 
We currently feel that the non-associativity
in $\mathbb{O}$ makes that algebra more difficult to use in a quantum mechanics 
frame work similarly as number theory is more difficult. 
But again, the Frobenius-Hurwitz theorem makes division algebras so
special that it is only natural to consider them. Particle physics today is in an interesting
stage, as some urgent questions are open. Are there more than three different flavors, are
there larger symmetries like $SU(5)$ or will completely different approaches
prevail and replace the standard model? The structure outlined here hints that the answer
could be ``no" in all three cases. 

\paragraph{Stability of neutrini}
If $a+ib$ or $a+jb$ or $a+kb$ is a Gaussian prime for which $a^2+b^2$ is prime then
it is also a Lipschitz prime in $\mathbb{H}$. While $\mathbb{C}$ contains also
``light primes" like $z=3$, primes for which $\sqrt{N(z)}$ is prime,
within $\mathbb{H}$: the initial inert prime $p=3$ would decay
$3 = (3,0,0,0) = (1,-1,0,-1) (1,1,0,1)$. The particle $(1,1,0,1)$ is equivalent to $(1,1,3,1)/2$
so that it is a meson. Why does $(3,0,0,0)$ not decay but
the prime $5=(1,2,0,0) (1,-2,0,0)$ does decay into an electron-positron pair? One possibility
to explain this is that $(3,0,0,0)$ is a different particle, evolving in a complex plane
rather than the meson pair $(1,-1,0,-1),(1,1,0,1)$ which evolves in the full algebra $\mathbb{H}$. 
However, here the stretch of the analogy becomes apparent already. 

\paragraph{Tetra and Penta quarks}
{\bf Tetra quarks} and {\bf Penta quark matter} seem have 
been confirmed by now. Examples of such events were X(3872), X(4140). 
In one interpretation, these states of matter so special and are just pairs of mesons
similarly as the {\bf deuterium nucleus} is a familiar 6 quark state and the tritium nucleus 
a $9$ quark state. Four and five quark states exotic because they are unusual 
mesons configurations which seem just be more rare. It is not necessary that a 4 
quark state contains all the quarks together in the same way as in say the proton or
neutron case. A deuterium is a 6 quark state, but the quarks are not glued as fundamentally
as in a baryon. A baryon can only be separated by combining it with other
particles, while a Deuterium boson can be split into a neutron and proton, which are both
Fermions.

\paragraph{Number of electrons and neutrinos}
We know that there are about an equal amount of $(4k+1)$- than
$(4k-1)$-primes, with a slight {\bf Chebyshev bias} for smaller numbers. But we estimate about 1 billion
times more neutrini than electron-positron pairs in the universe. How come?
Let us assume the number of particles in the universe is finite. As there are infinitely many primes
we have to assume we impose some {\bf Kustaanheimo threshold} \cite{Kustaanheimo} number and look only at
primes below that number and take this as a measure for the number of particles.
If we count also the {\bf virtual electron-positron pairs} which quantum field theory predicts
and leads to manifestation like the {\bf Casimir effect}.
Using virtual particles one could imagine the two particle classes to have similar amount of elements.

\paragraph{Light neutrini}
We know that neutrini are much lighter than electrons. The neutrino mass is about $0.33$ eV
which is about $5 \cdot 10^{-37} kg$. The mass of an electron is about $10^-{30} kg$. 
Current physics does not even tell how neutrinos acquire mass. The current understanding is that 
is that it does not happen through the Higgs mechanism. That would be compatible with the fact 
that it is simply related to the amplitude of the wave. Primes with larger norm have more mass. 
This is not unnatural since this is a dynamical mass interpretation. In number theory, there is a
difference in the amplitudes: the real inert Gaussian primes $z$ have prime $\sqrt{N(z)}$ while the 
split Gaussian primes $z$ have prime $N(z)$. For example $z=3$ is a inert prime wile $2+i$ is a
split prime. If we look at $N(z)$ as a measure for mass, then $N(z)$ is not prime for a real Gaussian
prime. It is the square root of $N(z)$ which is prime. If there should be any relation with $N(z)$
and mass and the mass relation measured is considered, 
this suggests that $\sqrt{N(z)} \sim 10^7$ for primes modeling particles. By Hurwitz,
on a sphere $N(z)=10^{14}$ this $N(z)+1$ is also the number of different prime equivalence classes 
on such a sphere. There is no relation with \cite{Kustaanheimo} since the later takes the
more radical approach replacing the field $\mathbb{C}$ with a finite field. 

\paragraph{The unramified primes}
The prime $p=2$ is special. When looking for analogies, it comes closest to the Higgs particle.
The reason is that this particle is conjugate to its own anti-particle and therefore has no charge.
The Higgs is a Boson unlike the neutral neutrini. What is the difference? 
One could try the picture that $2$ is actually 
equivalent to a pair of Fermions $z \overline{z}$ of norm $2$, where $z$ and $\overline{z}$ are 
conjugated. Ramified as such, the two 
Fermions $z,\overline{z}$ are identified so that it becomes elementary and a Boson. For the other 
primes like $3$, there is no such factorization. For $5$ the factorization $5=(2+i)(2-i)$ splits into two
charged particles which are not equivalent. 
The commutation statistics for $p=2$ is described by the {\bf second supplement of quadratic reciprocity}:
$(2|p) = (p|2)$ for primes of the form $8k \pm 1$ and $(2|p) = -(p|2)$ for the rest of odd primes.
We can remember this as follows: while neutrini do not interact with electrons nor positrons,
the prime $2$ interacts with both types.

\paragraph{What are the photons in this allegory?}
Photons are the units in $\mathbb{C}$. They play the role
of gauge transformations. The photons $i,-i$ satisfy $i^2=-1$ and interact with electrons or
positrons only and not with neutrini. Indeed, by the first supplement of quadratic reciprocity, 
we have $(-1|p)=1$ for electron for positrons primes but $(-1|p)=-1$ for neutrini primes: 
there is no solution to $x^2=i^2$ modulo $p$ if $p$ is a neutrino prime but there is a solution 
$x$ to $x^2=i^2$ modulo $p$ if $p$ is a charged lepton. 
All photons are gauge equivalent so that there is only one kind of a photon. 
A photon has zero charge as it is equivalent to the real element $1$. 

\paragraph{Seduced by analogies}
To the end, one has to caution: 
tales like Eddington explaining the {\bf fine structure
constant} by numerology or Kepler obsessed with his {\bf Harmonices Mundi} matching planetary motion
with Platonic solids or the Thomson-Tait approach
to explain the {\bf chemical elements} using knots, show that analogies do not 
always go well, even if done by eminent scientists.
As an {\bf amusing illustration}, in some online forum, a user discovered that
the natural numbers $\mathbb{N}$, the complex numbers $\mathbb{C}$, the quaternions $\mathbb{H}$ and the 
octonions $\mathbb{O}$ are curiously linked to the basic constituents of {\bf organic matter} like nitrogen 
$N$, carbon $C$, hydrogen $H$ and oxygen $O$. This silly coincidence illustrates how scientific brains like 
to build connections and if necessary, bend things, if a match is not complete. While the allegory 
spun here could be of similar type, there is interesting combinatorics
to be explored, as the figures below illustrate. How many different baryons are there for
a given prime? How many are there of each type? Is there a more natural way to associate charge to the
quarks by looking at all possible factorizations.
What exactly happens to the equivalence classes during meta commutation?

\section{Figures}

\begin{figure}
\scalebox{0.80}{\includegraphics{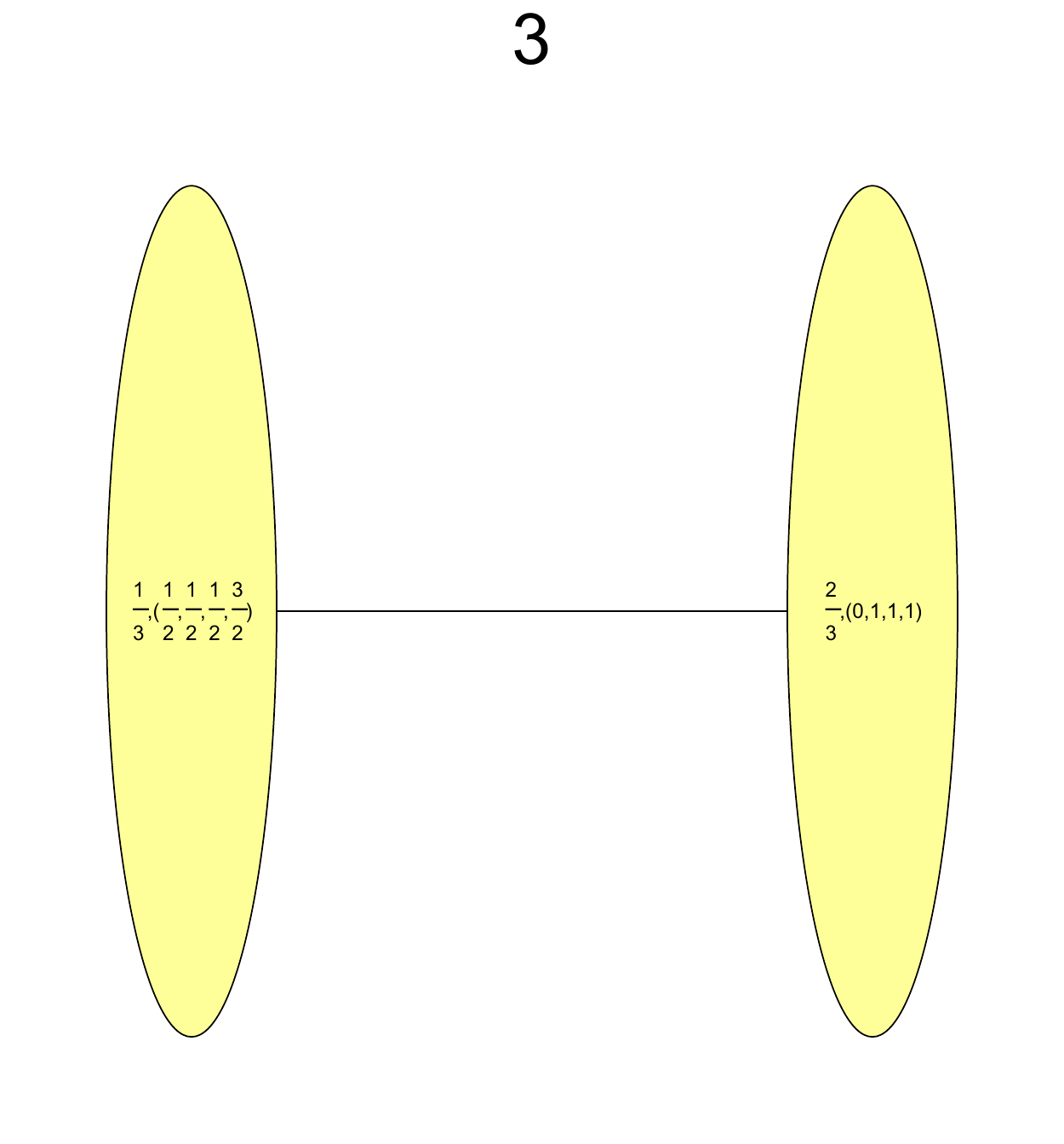}}
\scalebox{0.80}{\includegraphics{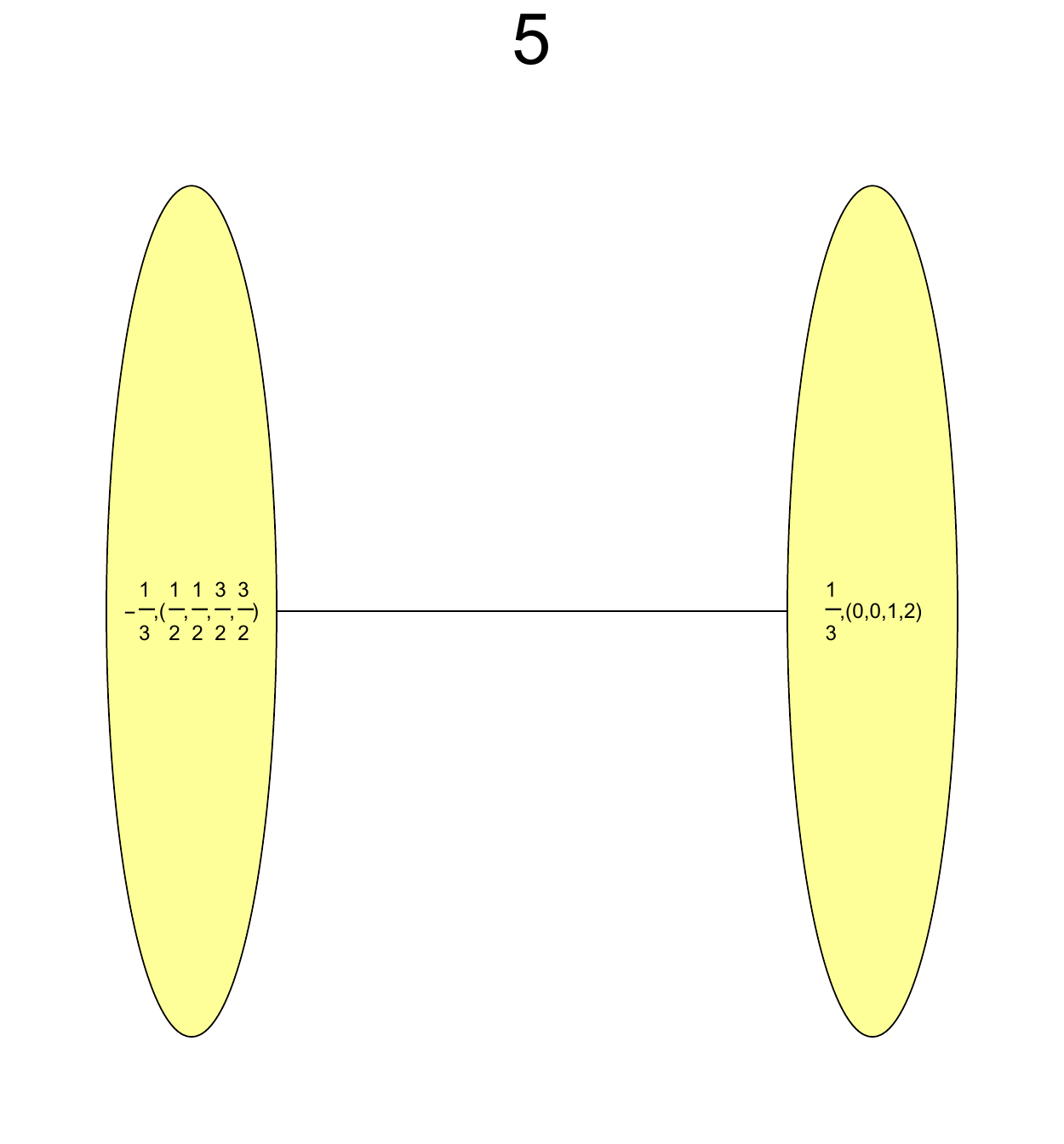}}
\caption{ \label{quarks} Hadrons with charges.  } \end{figure} 

\begin{figure}
\scalebox{0.80}{\includegraphics{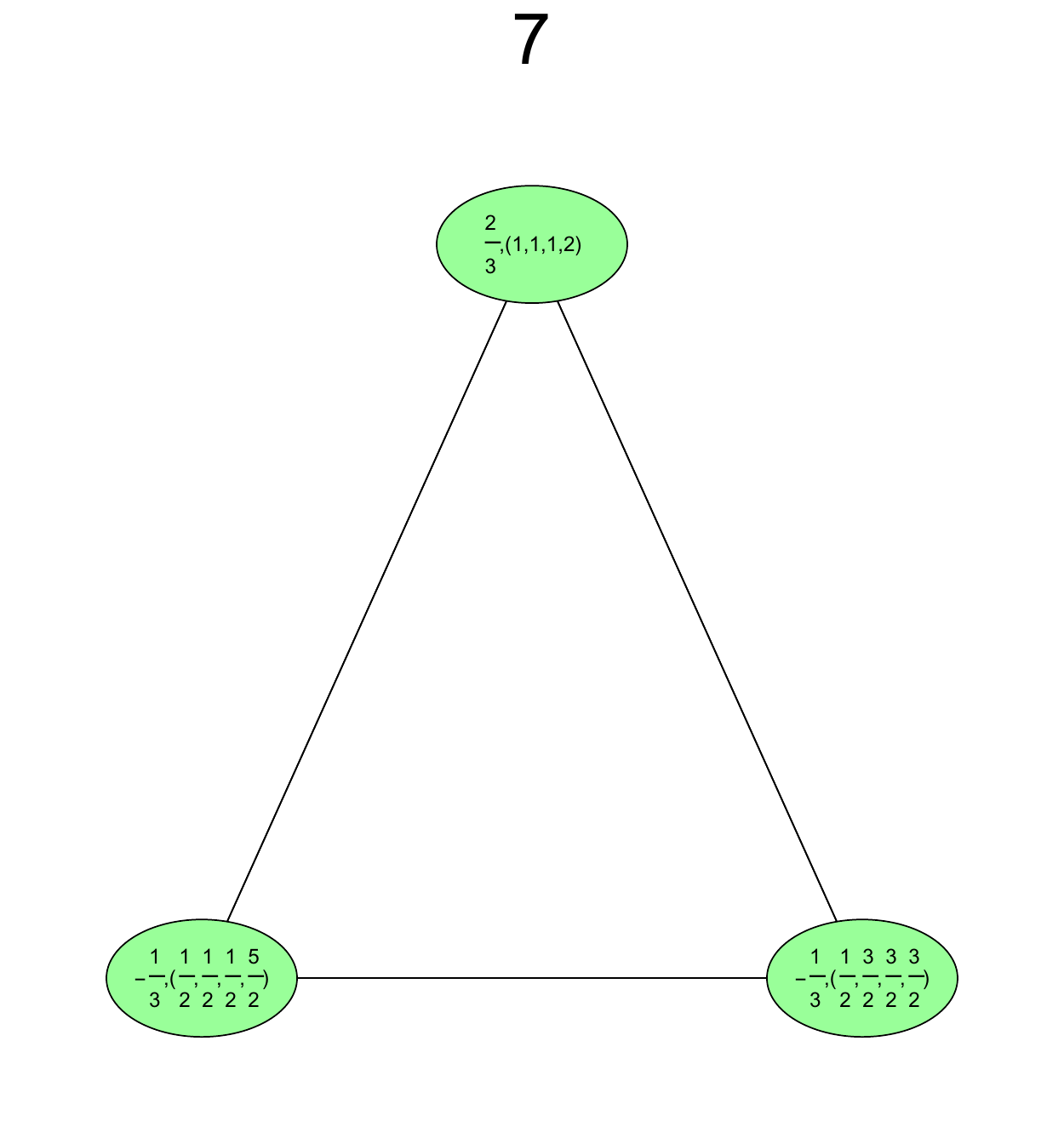}}
\scalebox{0.80}{\includegraphics{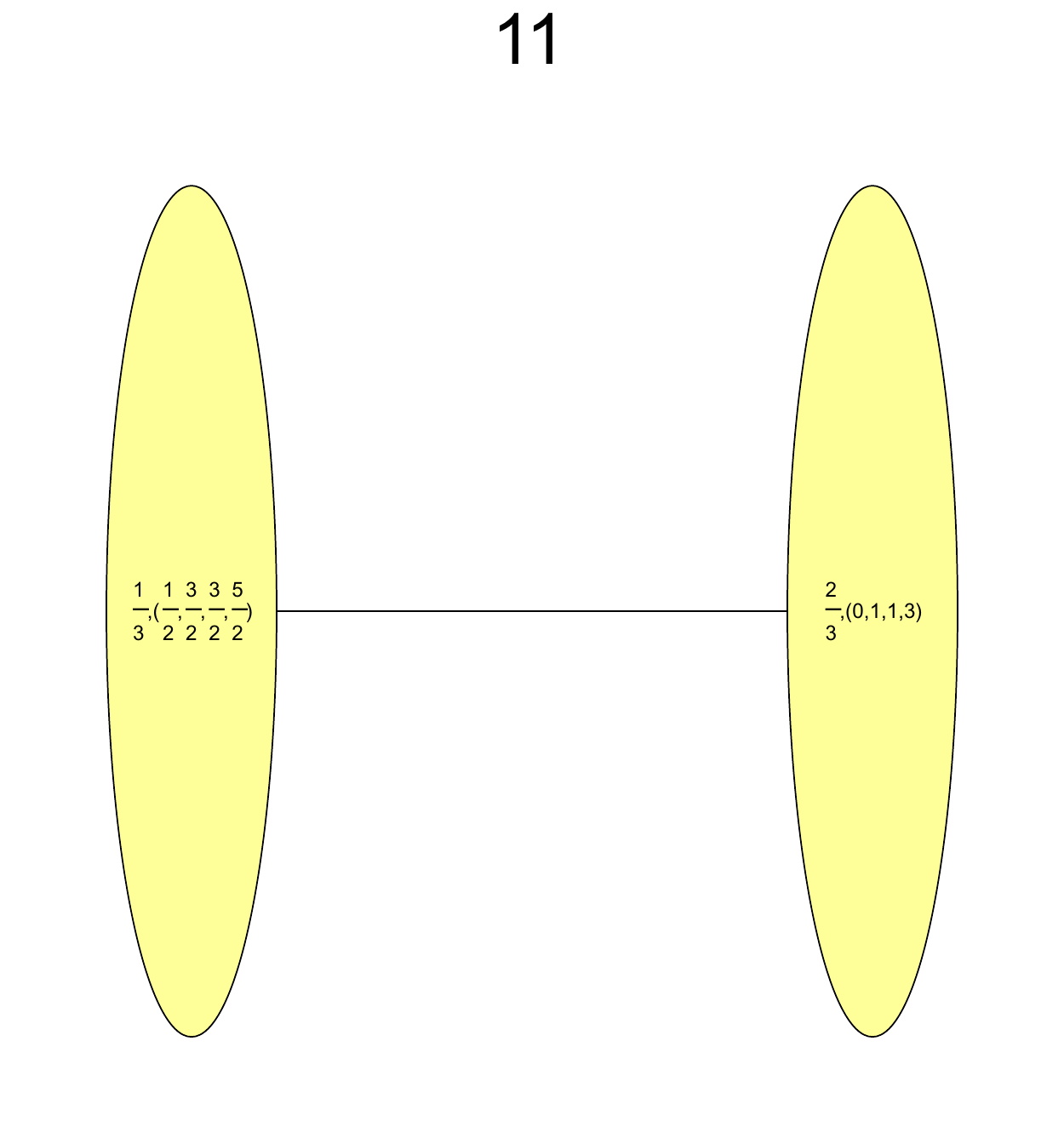}}
\caption{ \label{quarks} Hadrons with charges.  } \end{figure} 

\begin{figure}
\scalebox{0.80}{\includegraphics{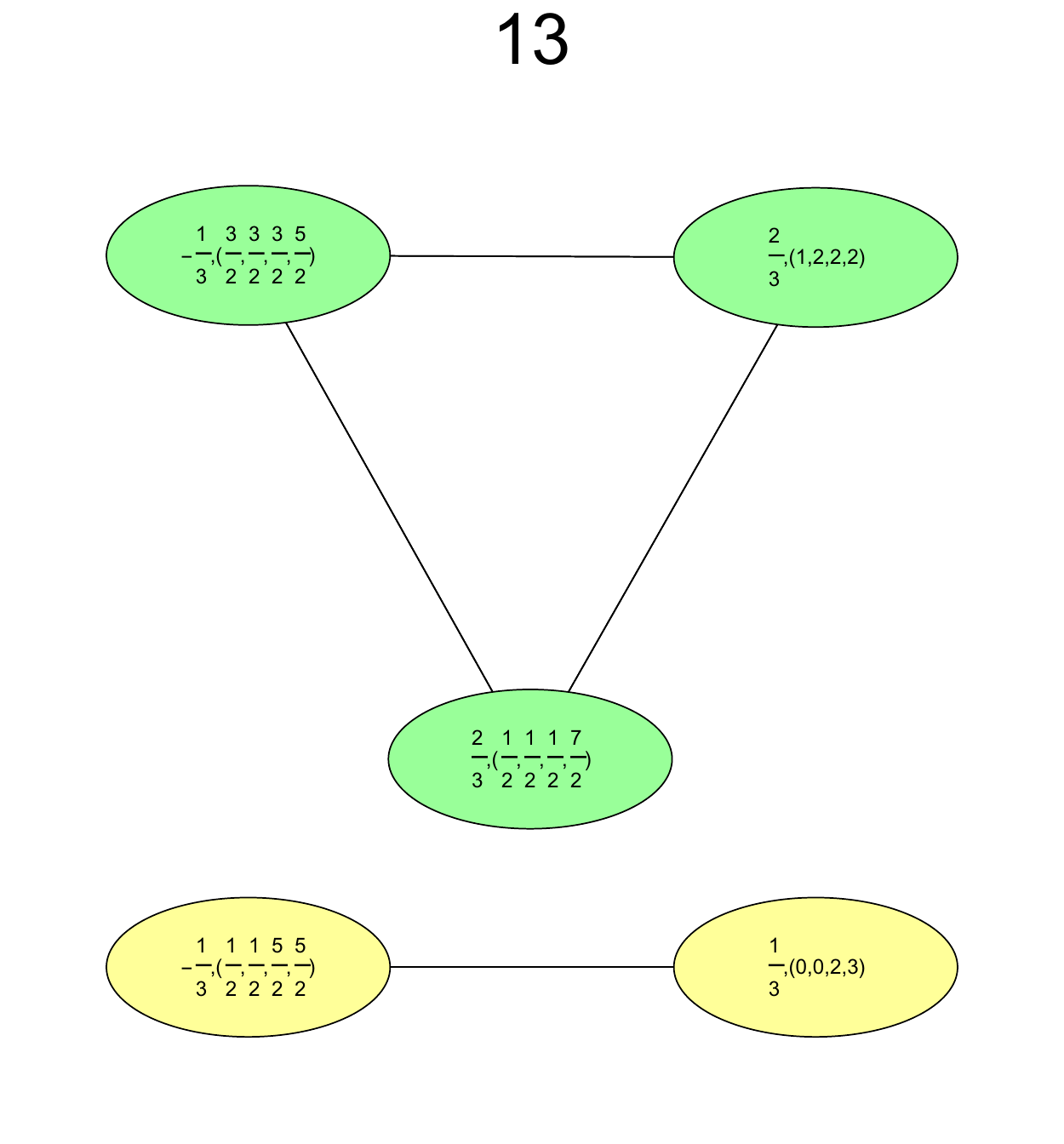}}
\scalebox{0.80}{\includegraphics{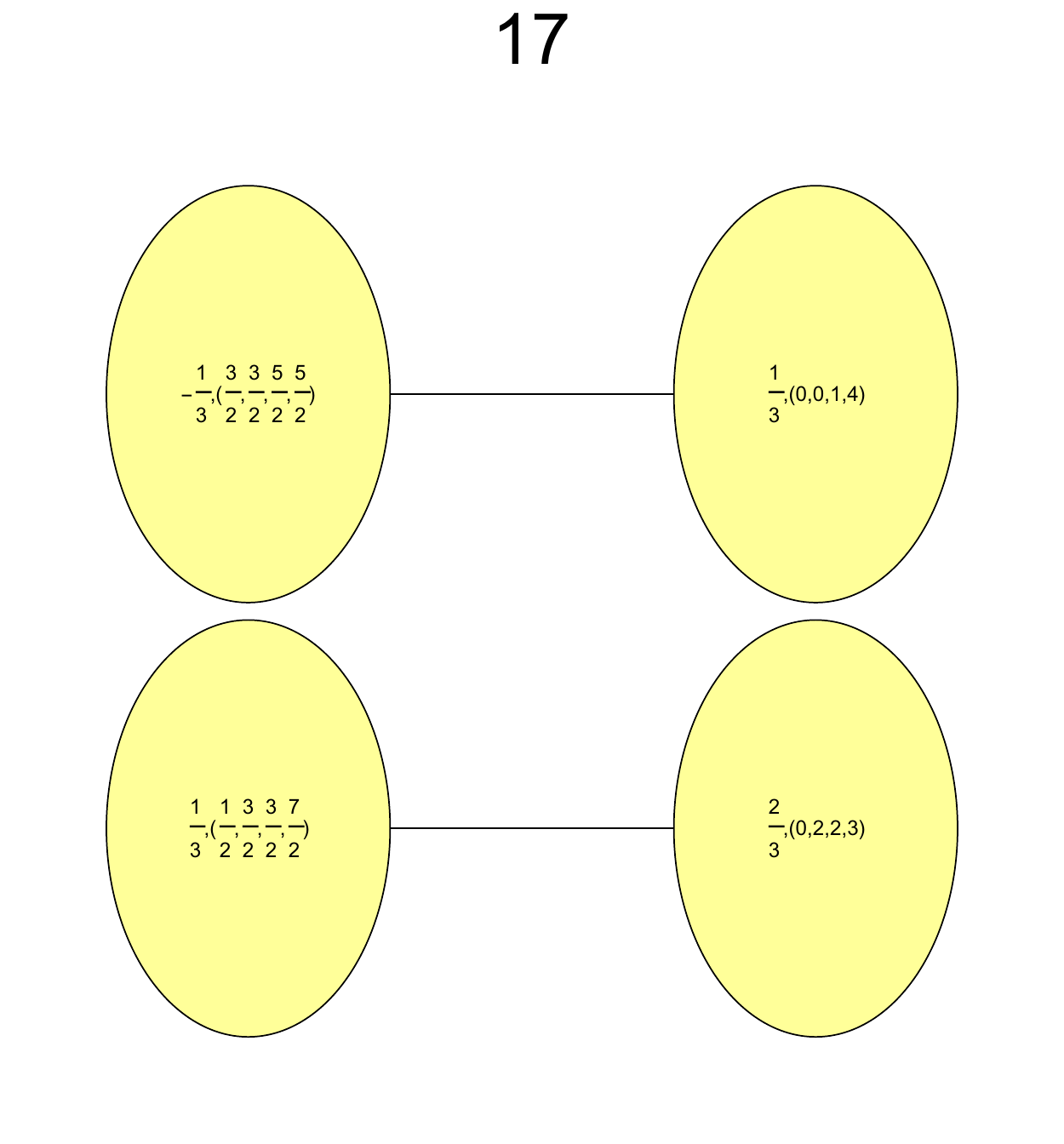}}
\caption{ \label{quarks} Hadrons with charges.  } \end{figure} 

\begin{figure}
\scalebox{0.80}{\includegraphics{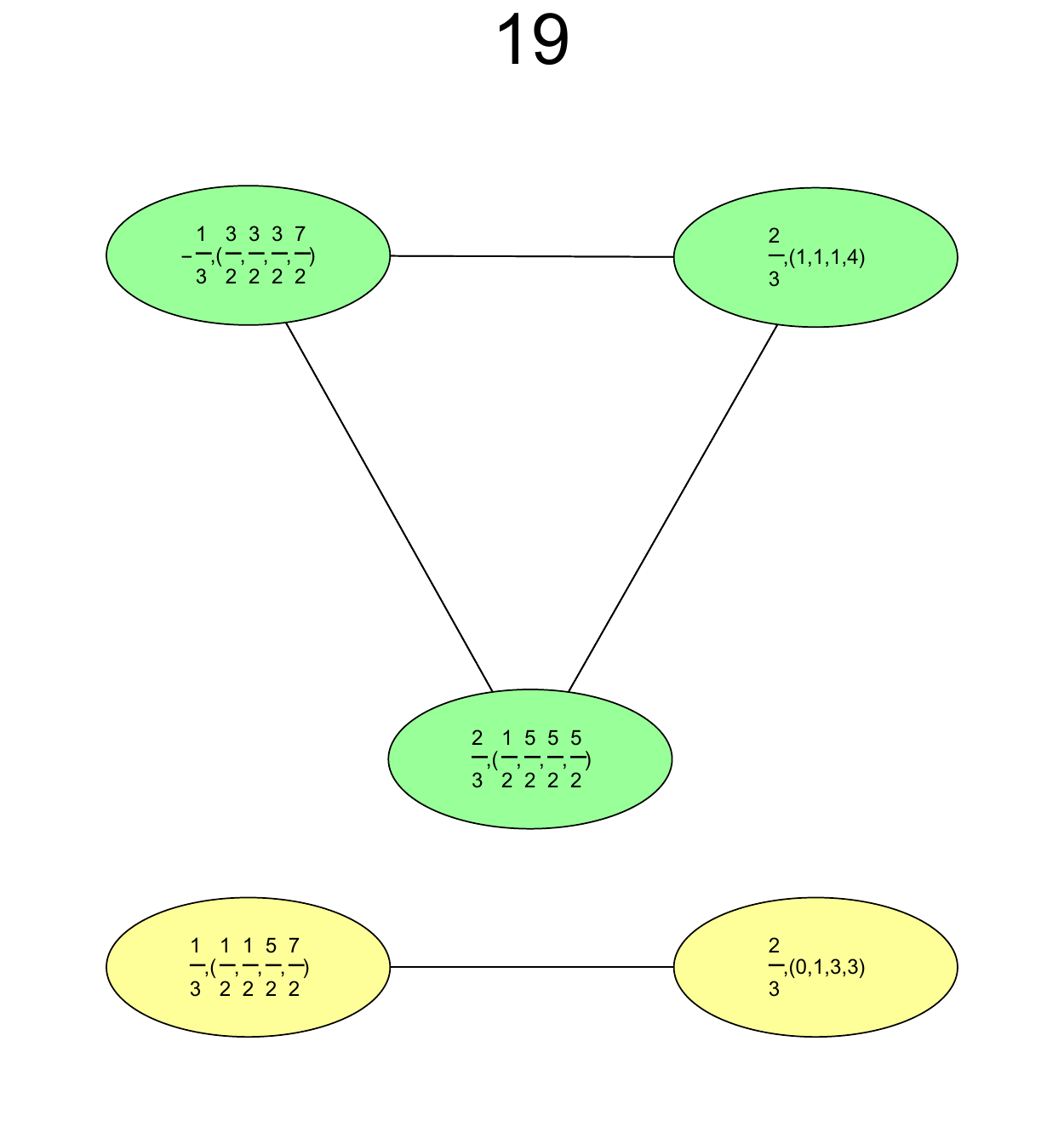}}
\scalebox{0.80}{\includegraphics{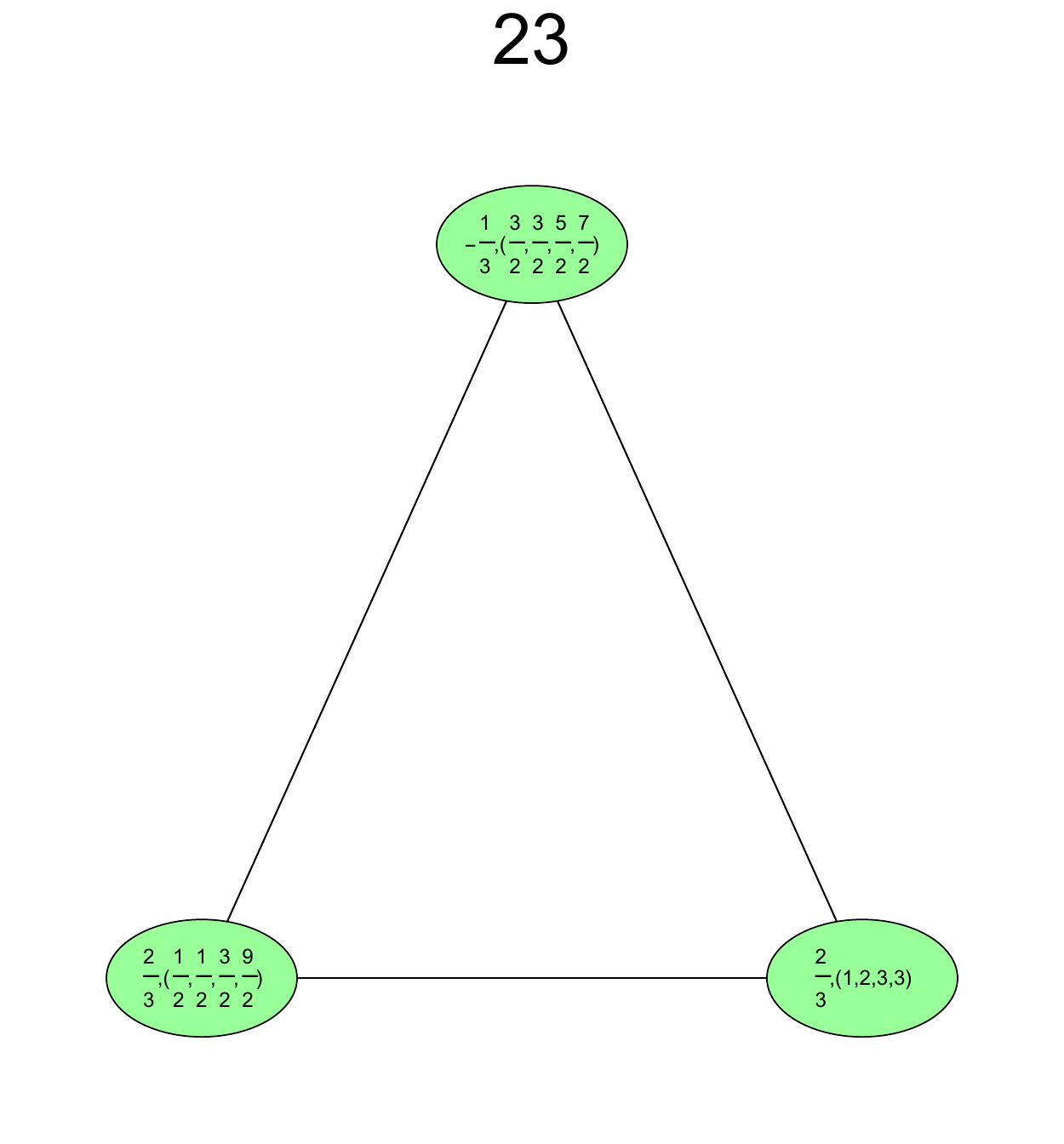}}
\caption{ \label{quarks} Hadrons with charges.  } \end{figure}

\begin{figure}
\scalebox{0.80}{\includegraphics{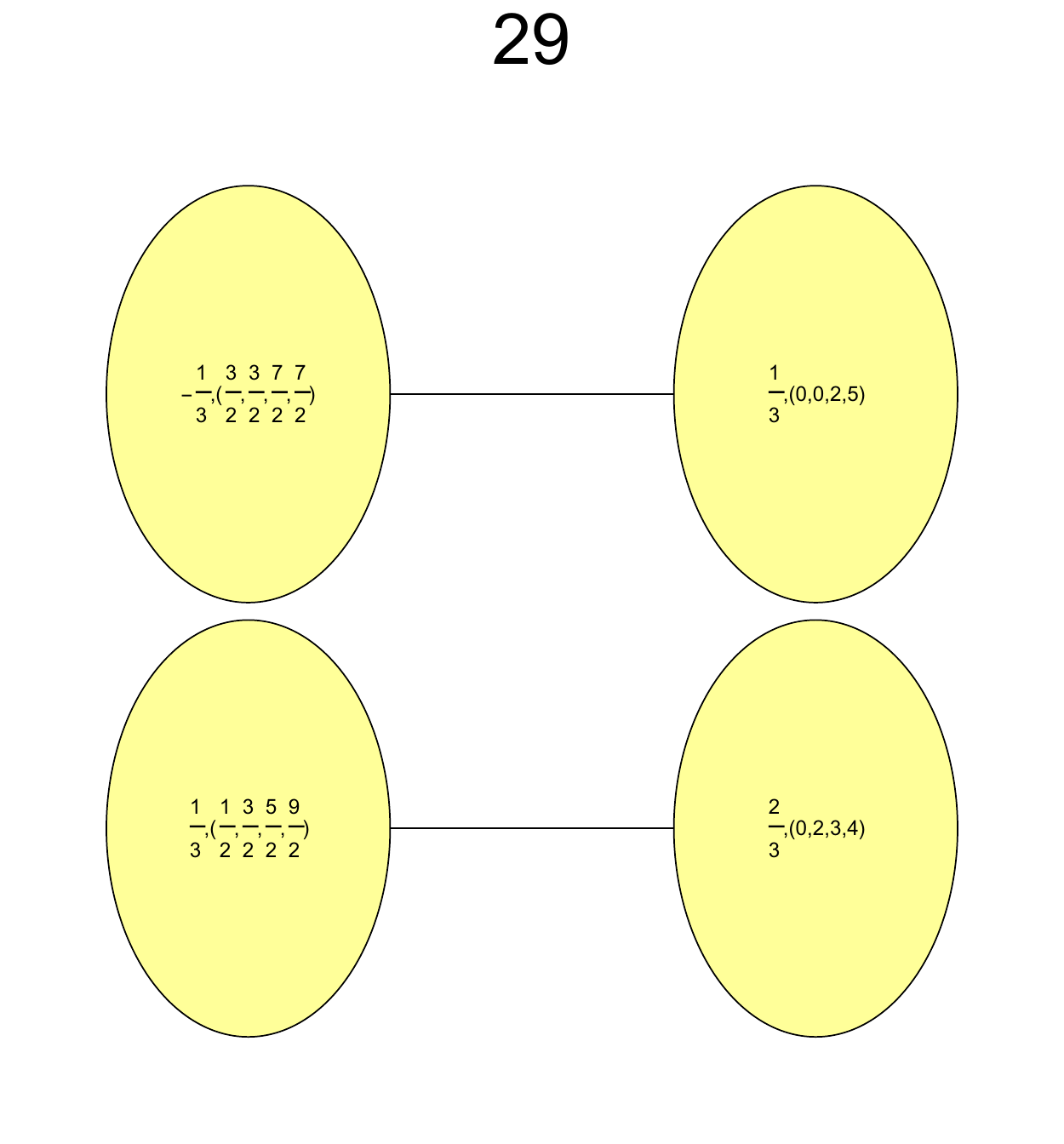}}
\scalebox{0.80}{\includegraphics{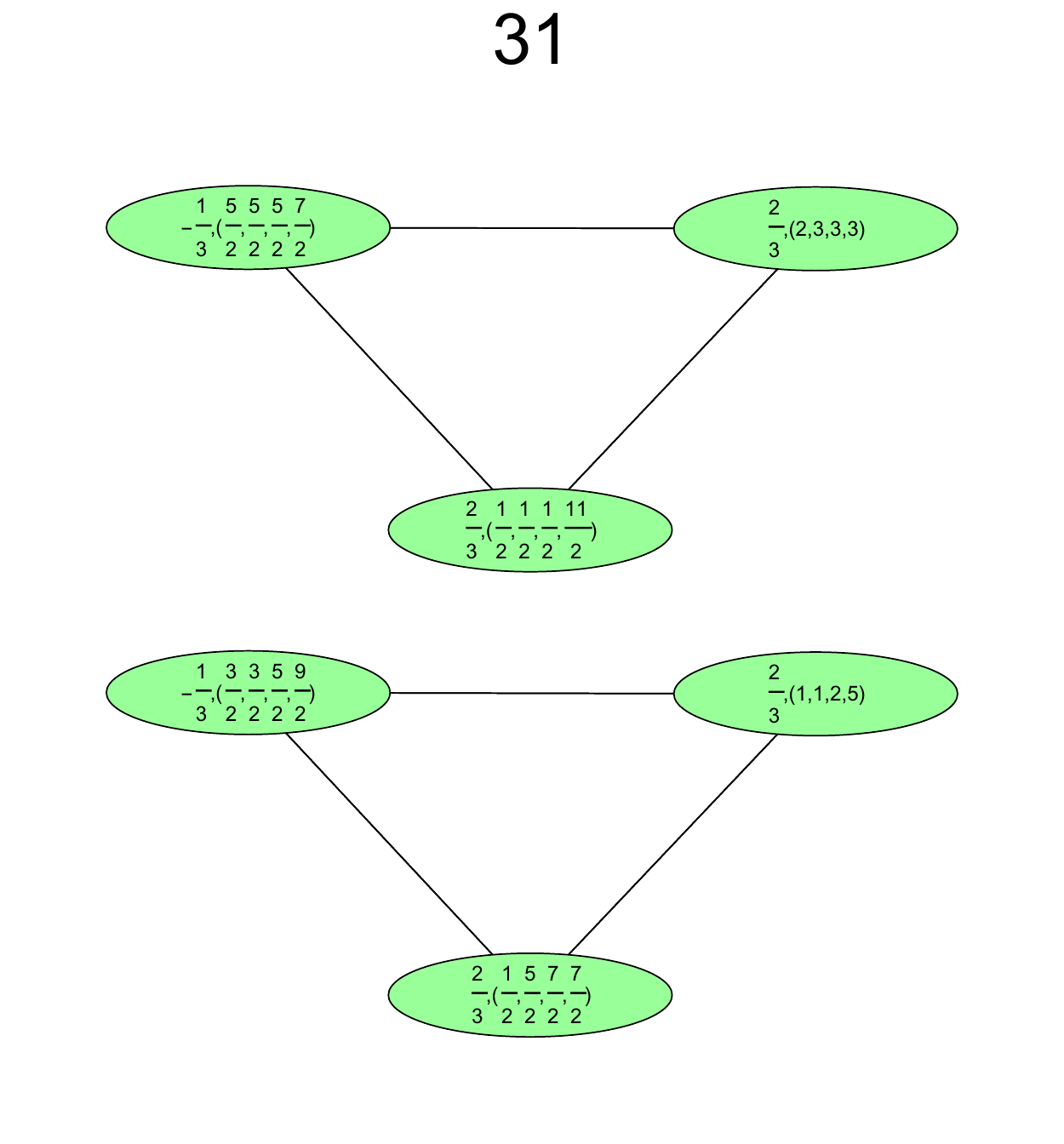}}
\caption{ \label{quarks} Hadrons with charges.  } \end{figure} 

\begin{figure}
\scalebox{0.80}{\includegraphics{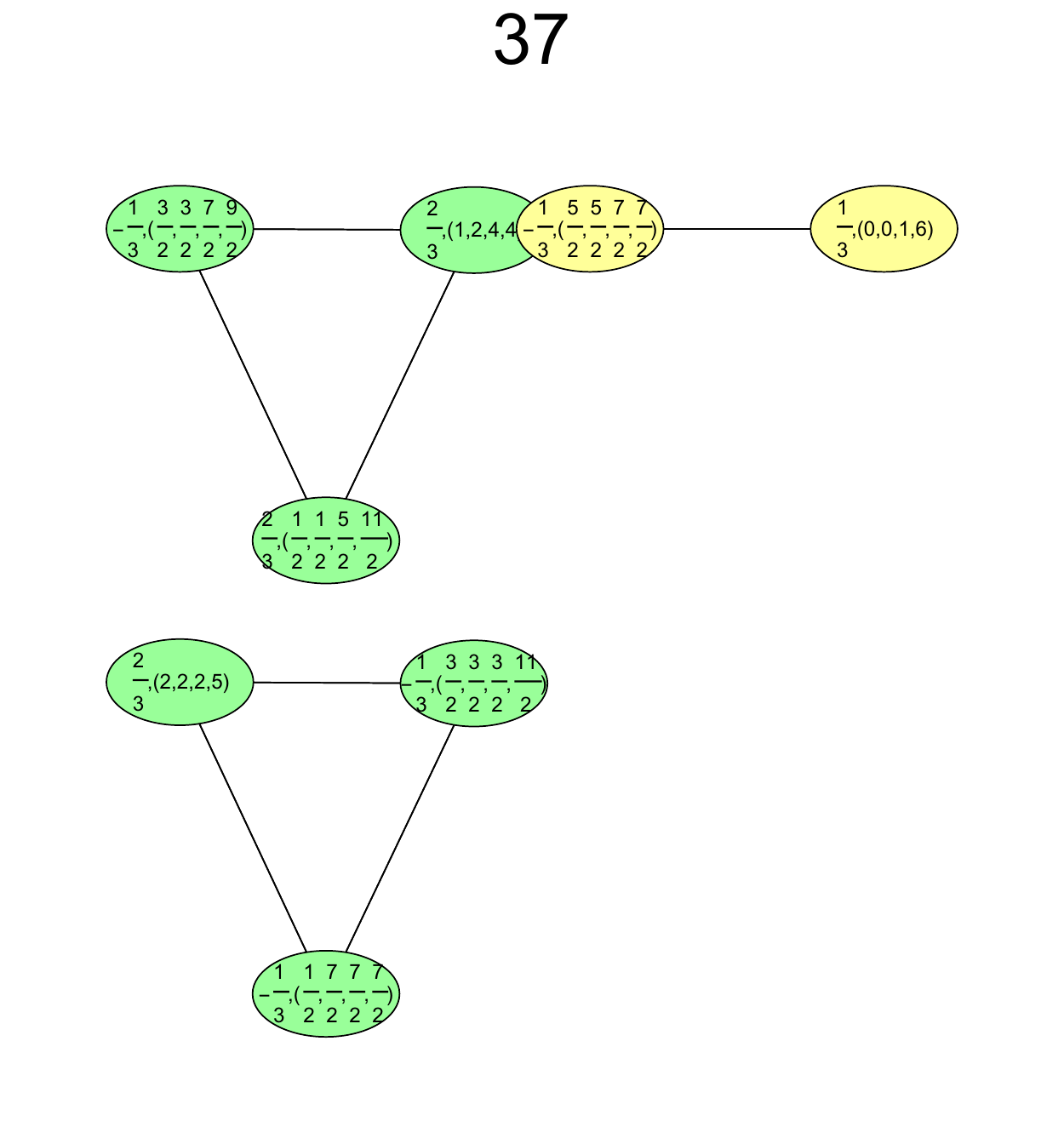}}
\scalebox{0.80}{\includegraphics{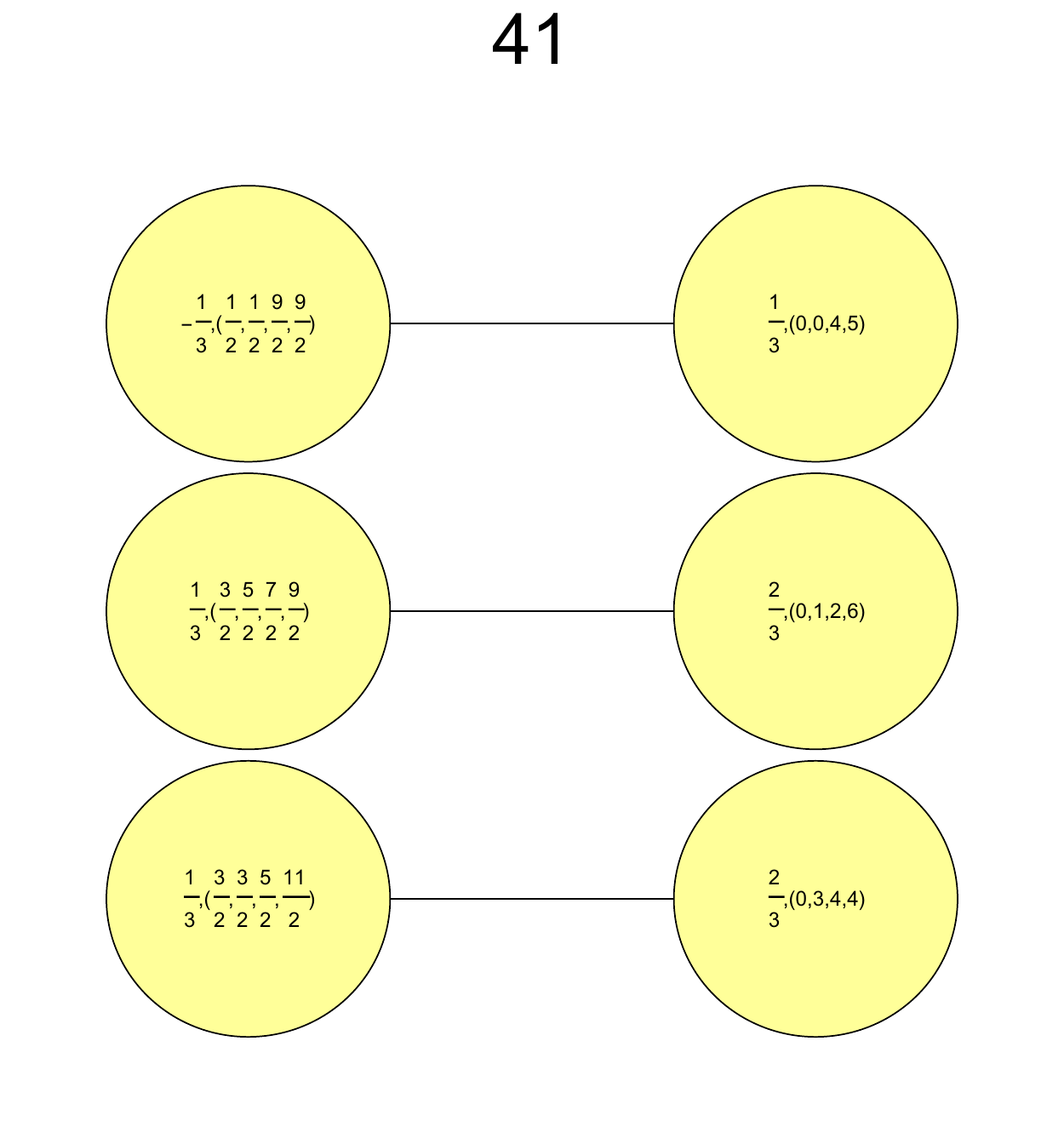}}
\caption{ \label{quarks} Hadrons with charges.  } 
\end{figure}

\begin{figure}
\scalebox{0.80}{\includegraphics{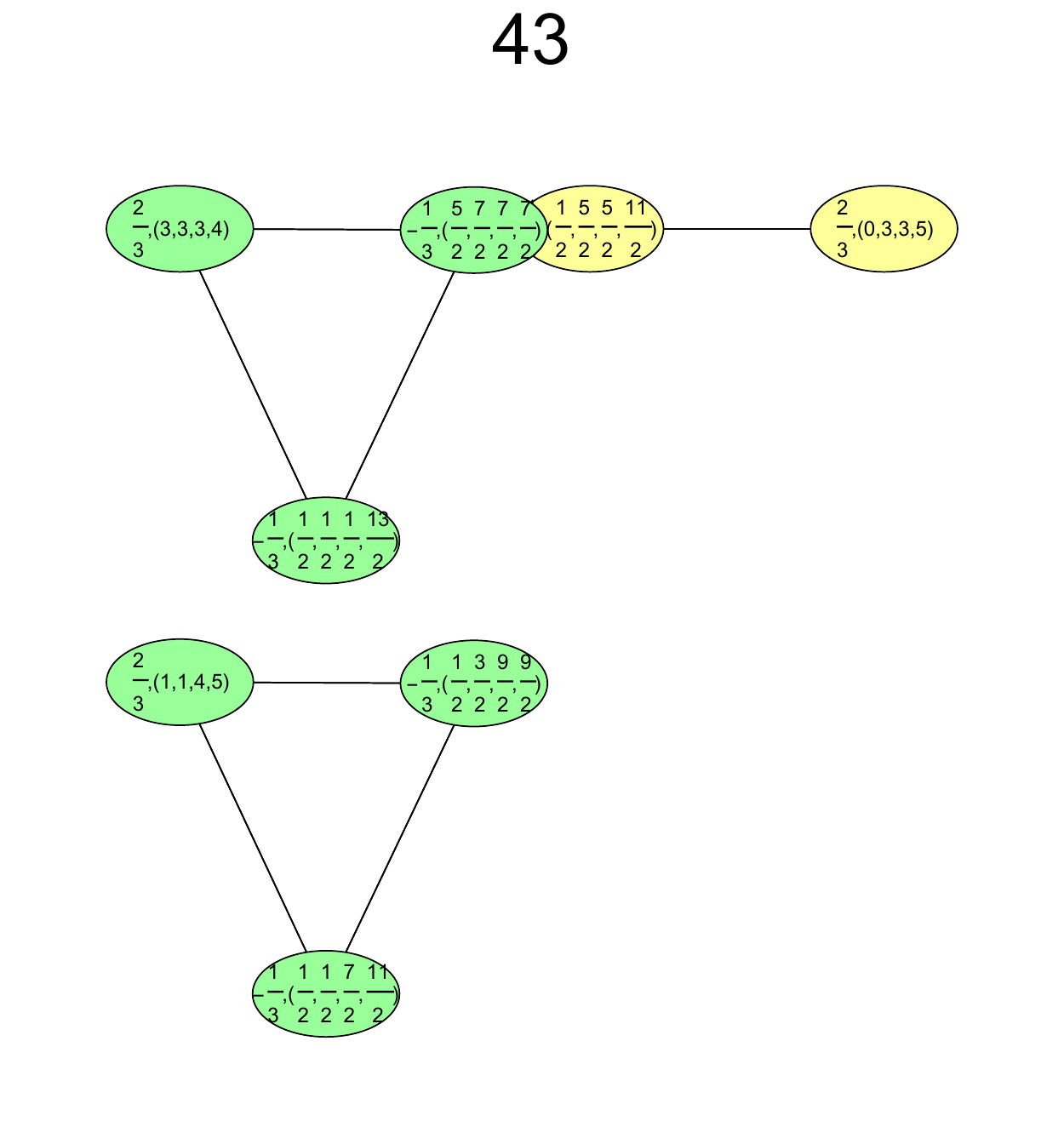}}
\scalebox{0.80}{\includegraphics{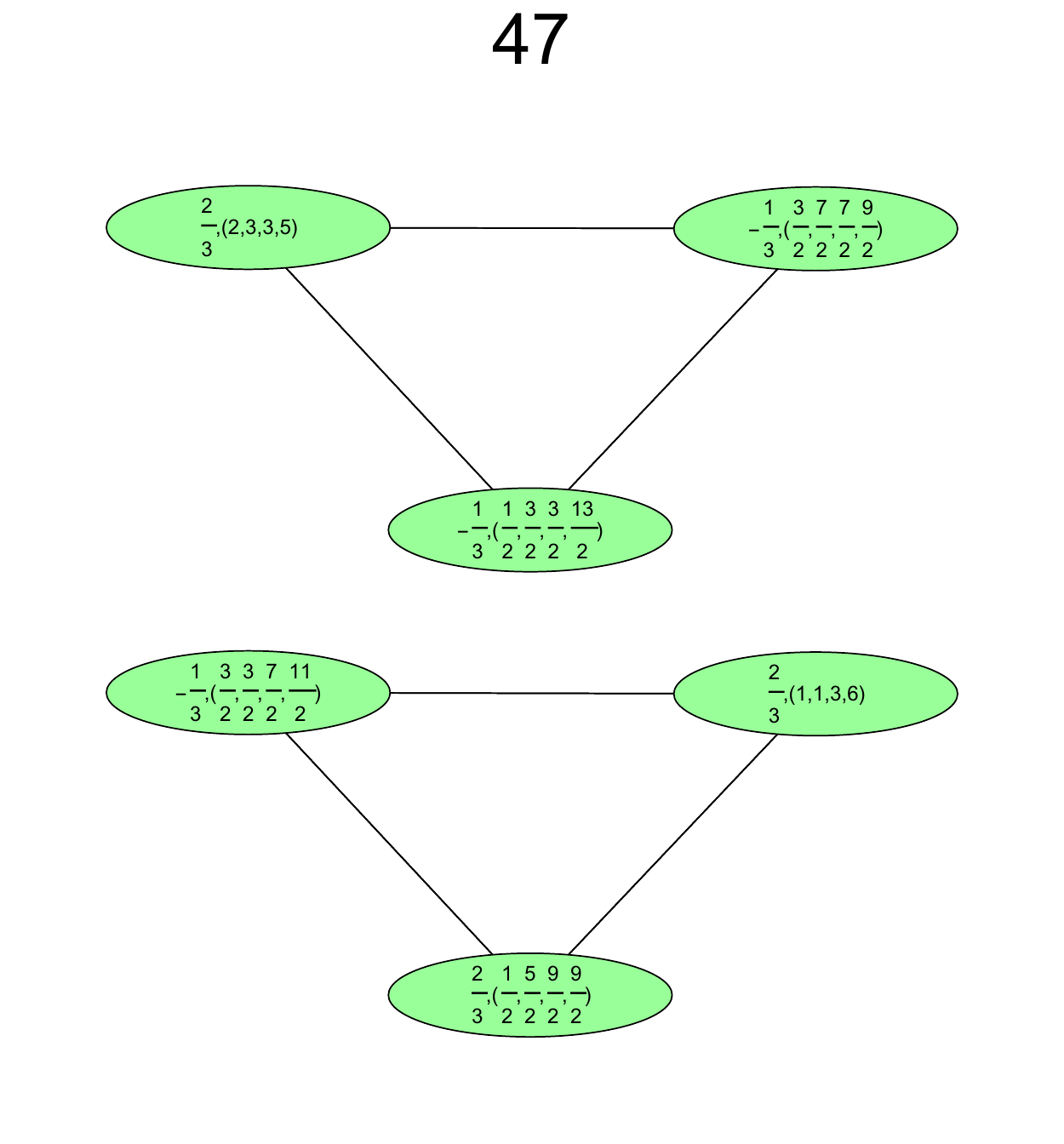}}
\caption{ \label{quarks} Hadrons with charges.  } \end{figure} 

\begin{figure}
\scalebox{0.80}{\includegraphics{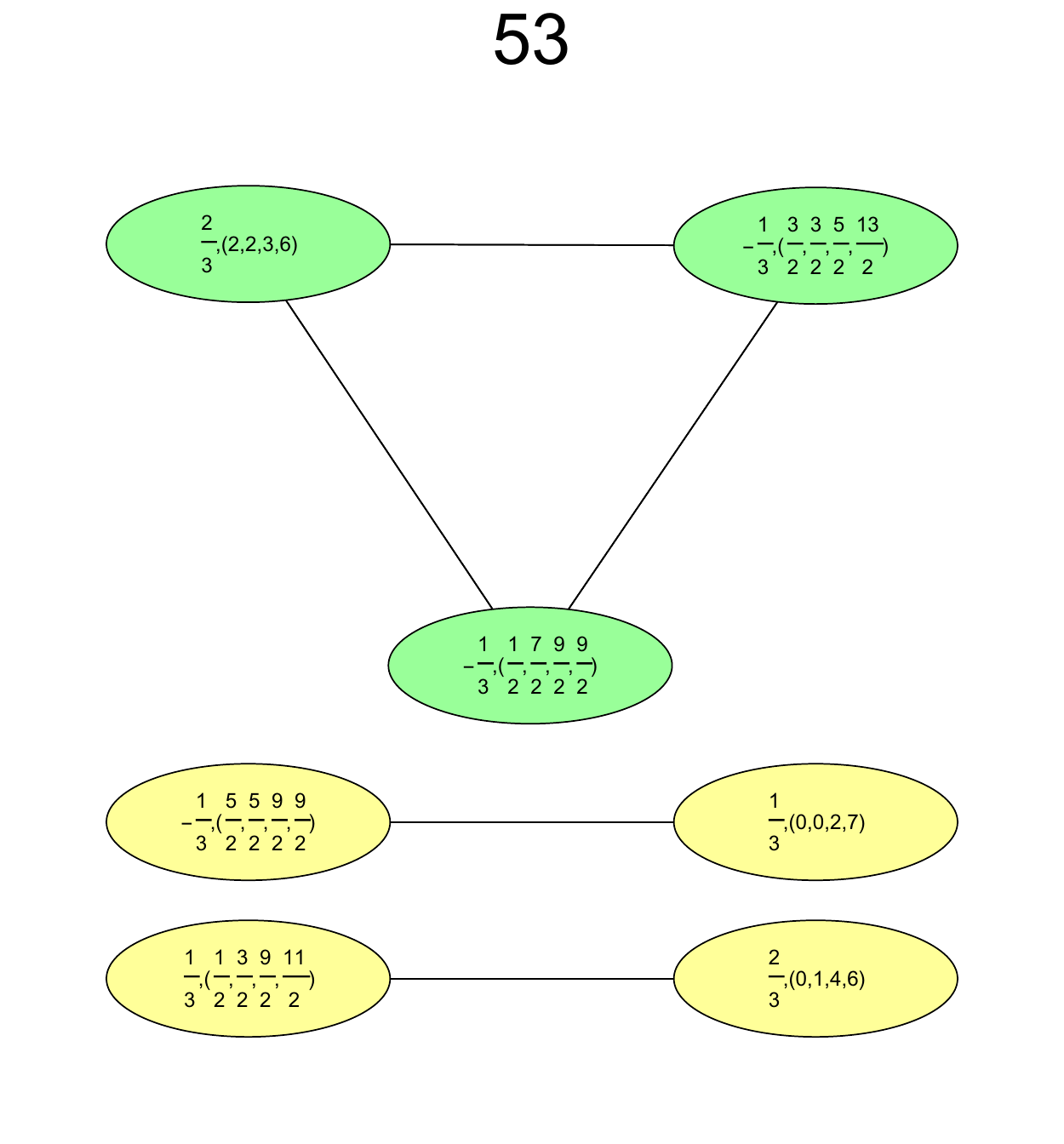}}
\scalebox{0.80}{\includegraphics{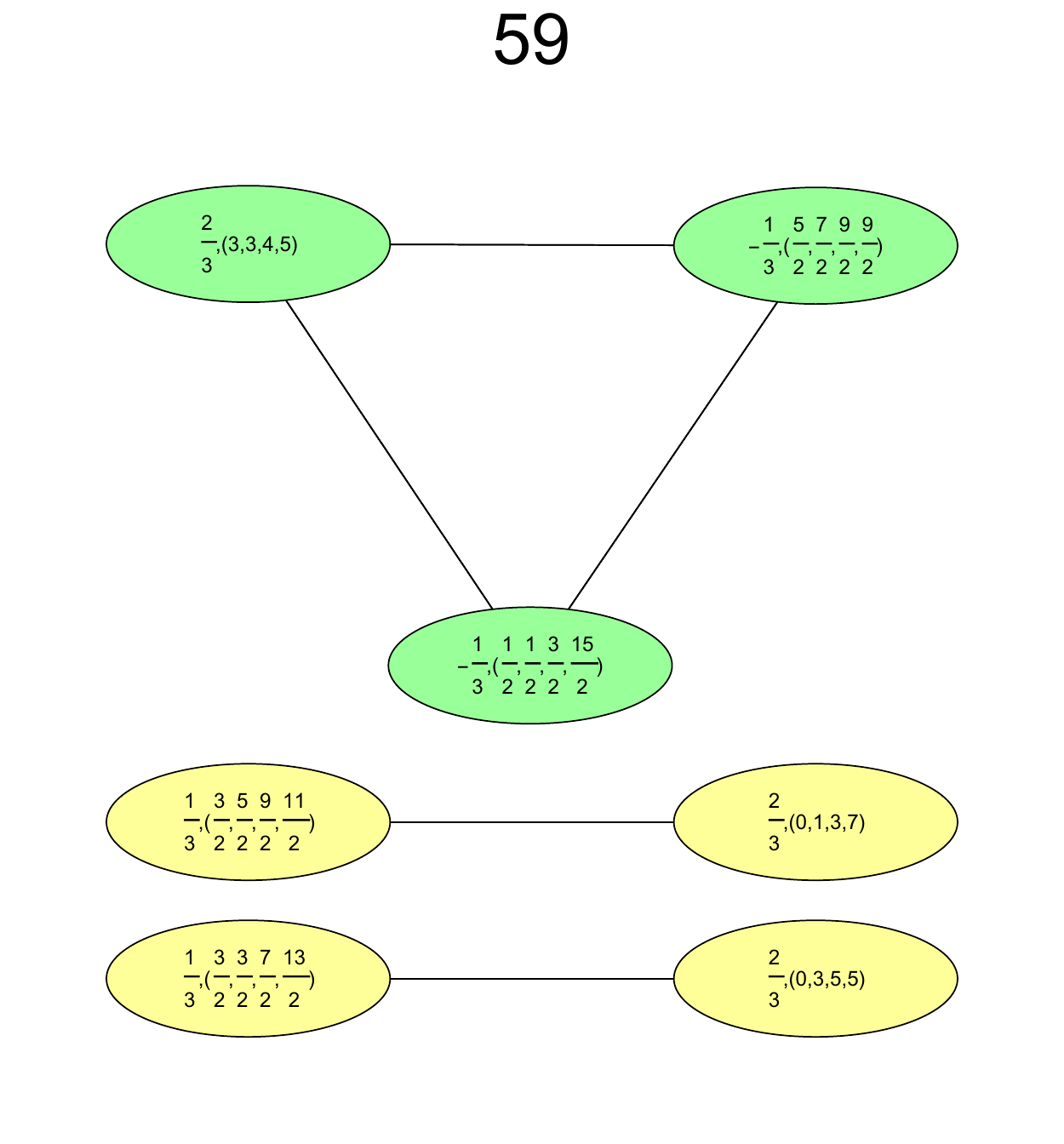}}
\caption{ \label{quarks} Hadrons with charges.  } 
\end{figure}

\begin{figure}
\scalebox{0.80}{\includegraphics{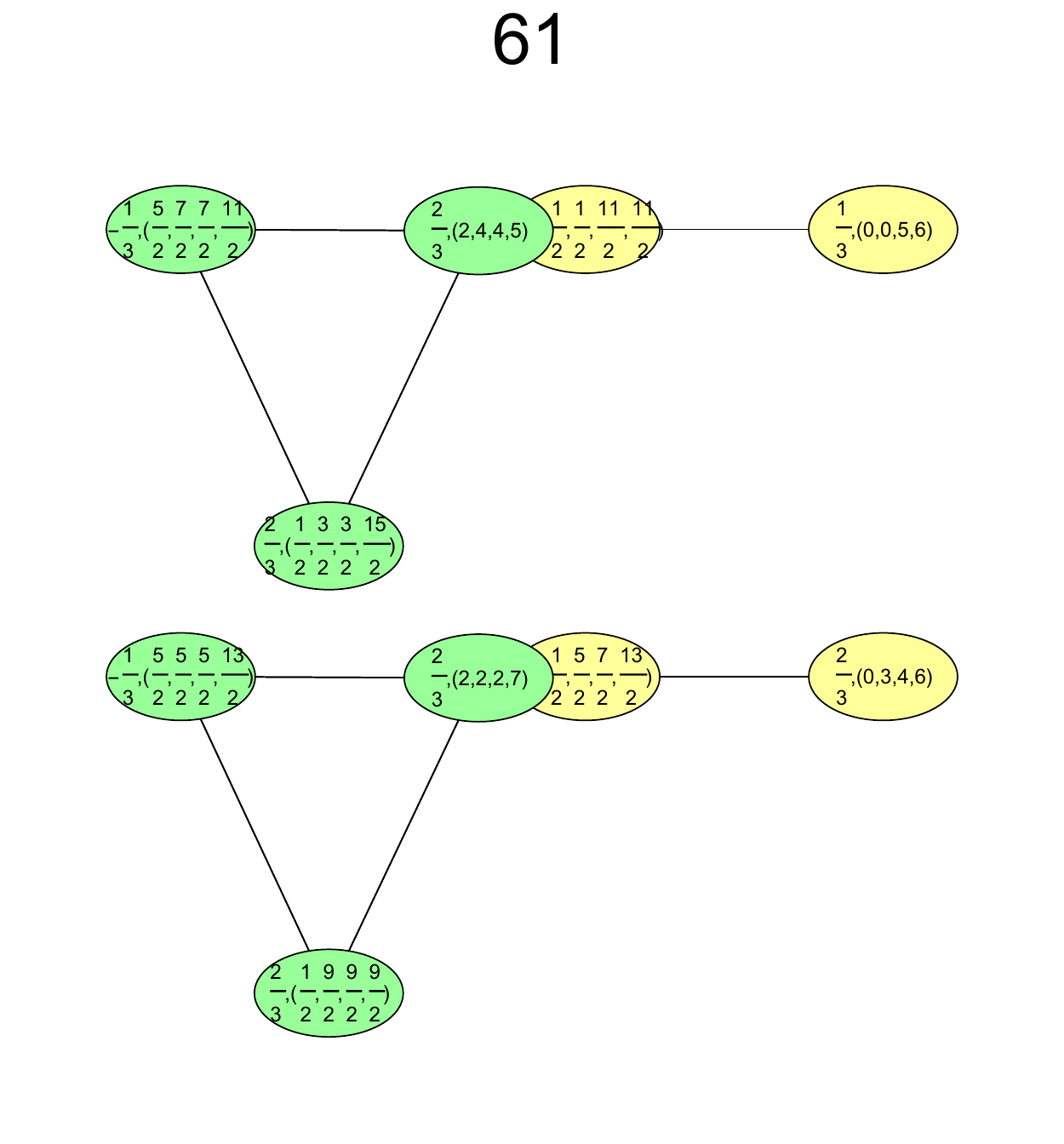}}
\scalebox{0.80}{\includegraphics{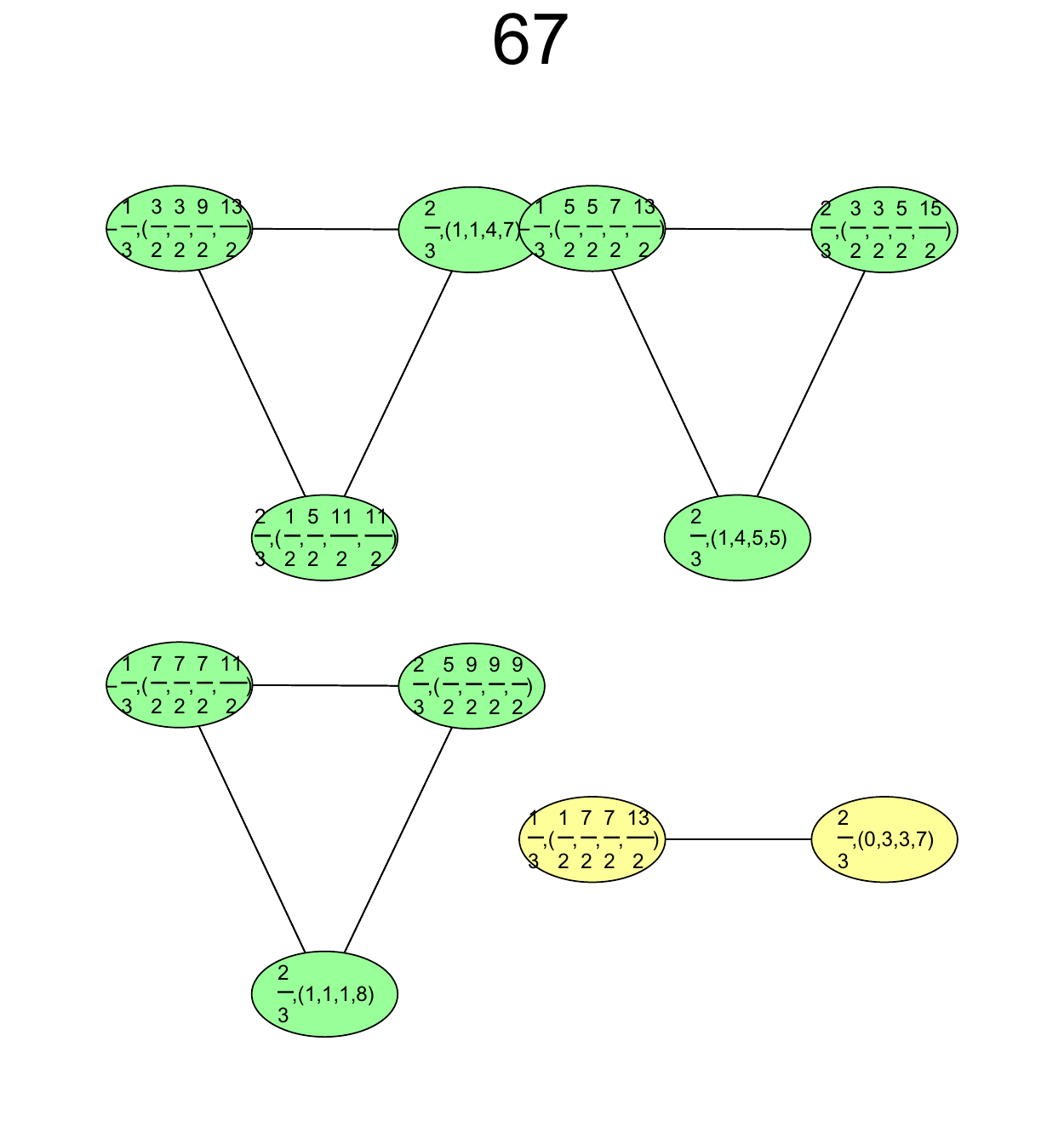}}
\caption{ \label{quarks} Hadrons with charges.  } \end{figure}

\begin{figure}
\scalebox{0.80}{\includegraphics{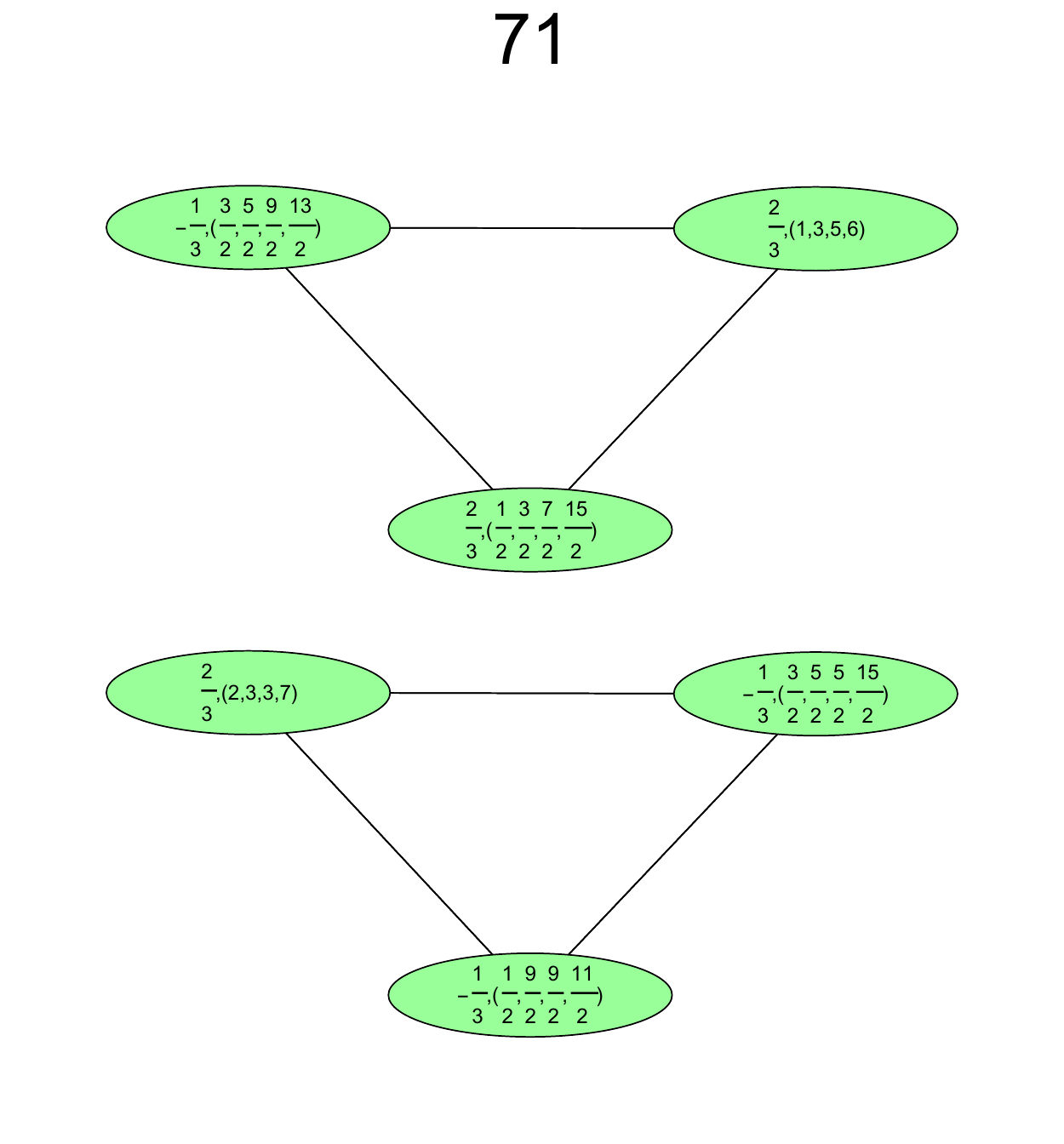}}
\scalebox{0.80}{\includegraphics{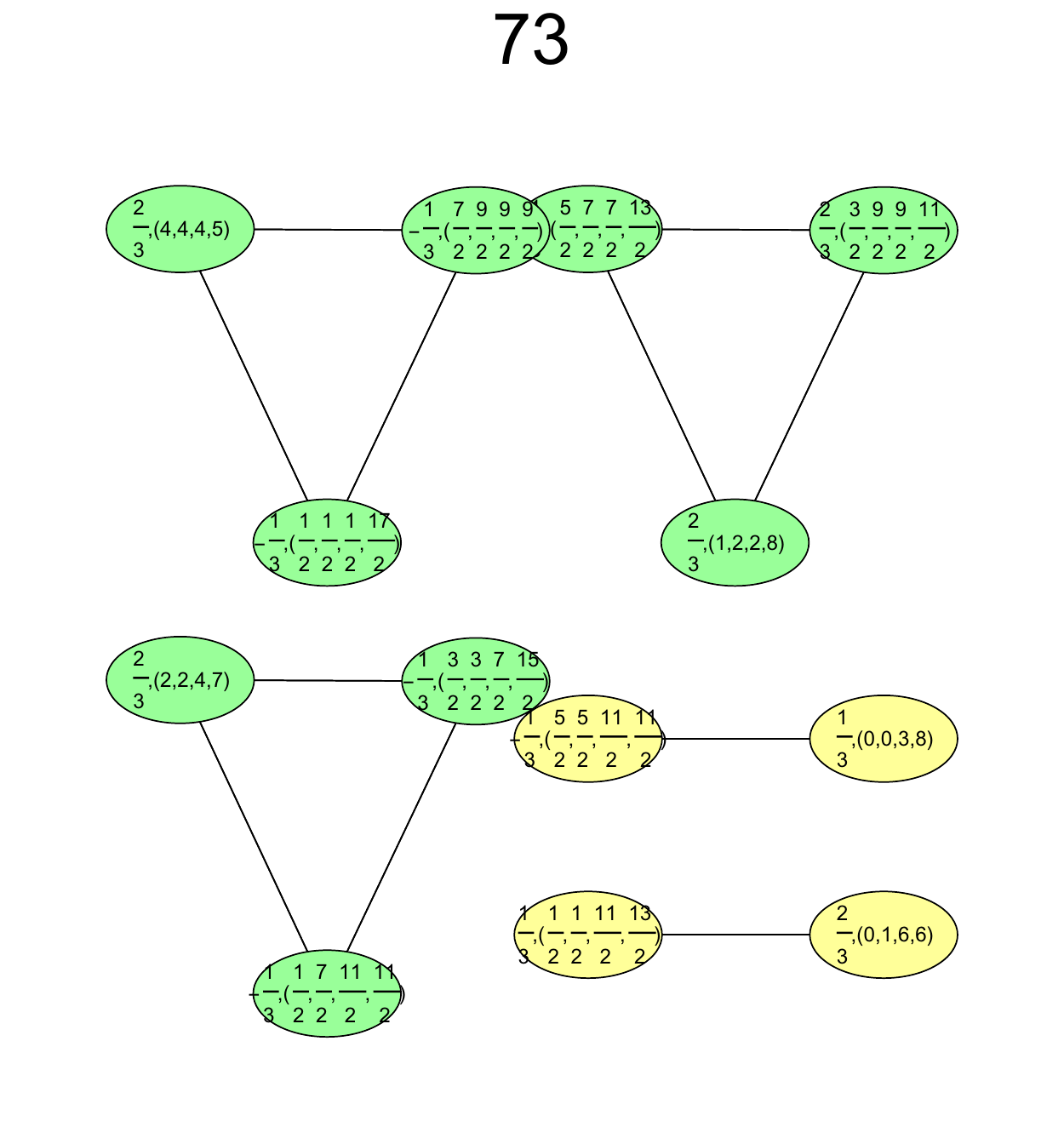}}
\caption{ \label{quarks} Hadrons with charges.  } 
\end{figure}

\begin{figure}
\scalebox{0.80}{\includegraphics{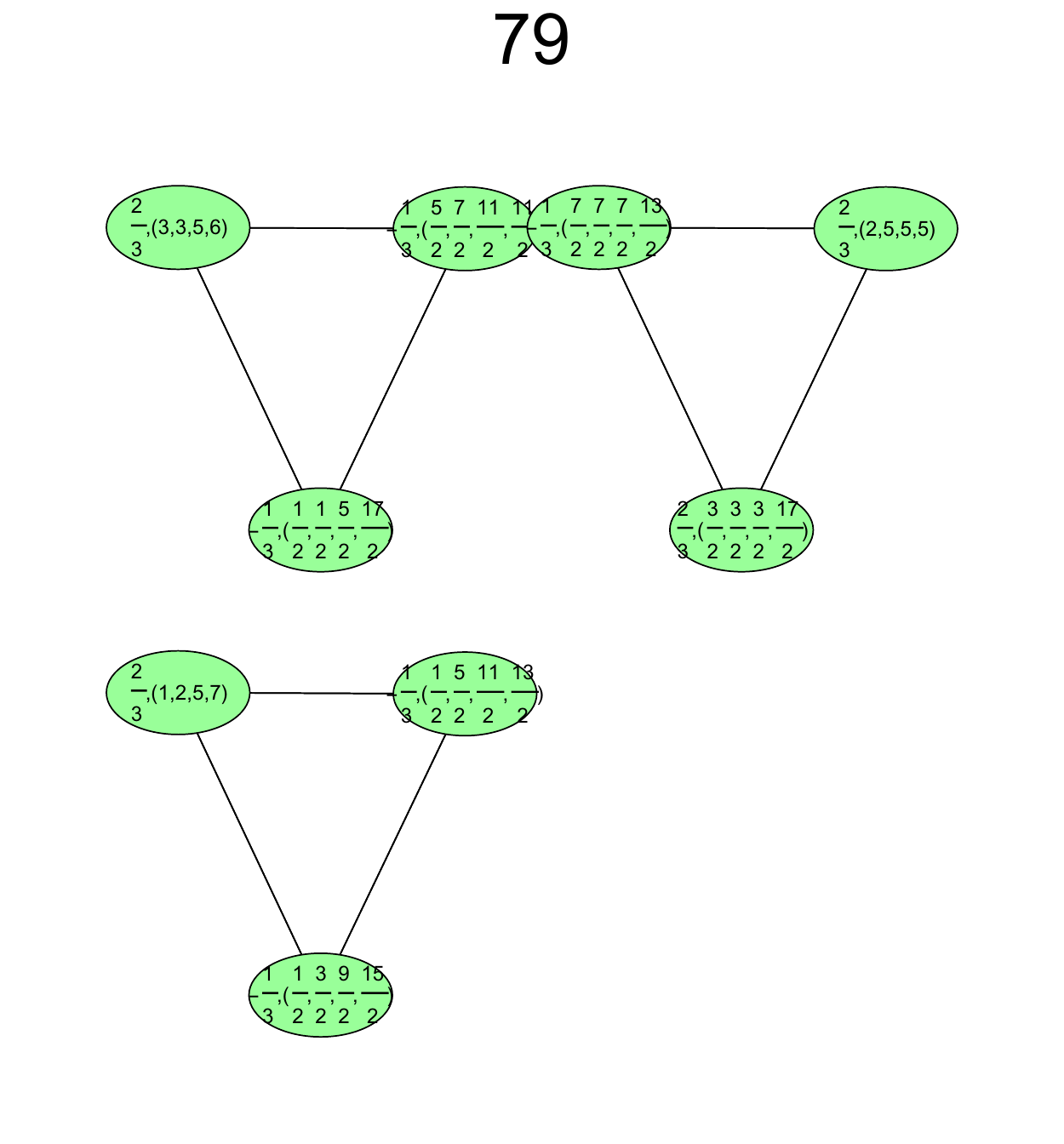}}
\scalebox{0.80}{\includegraphics{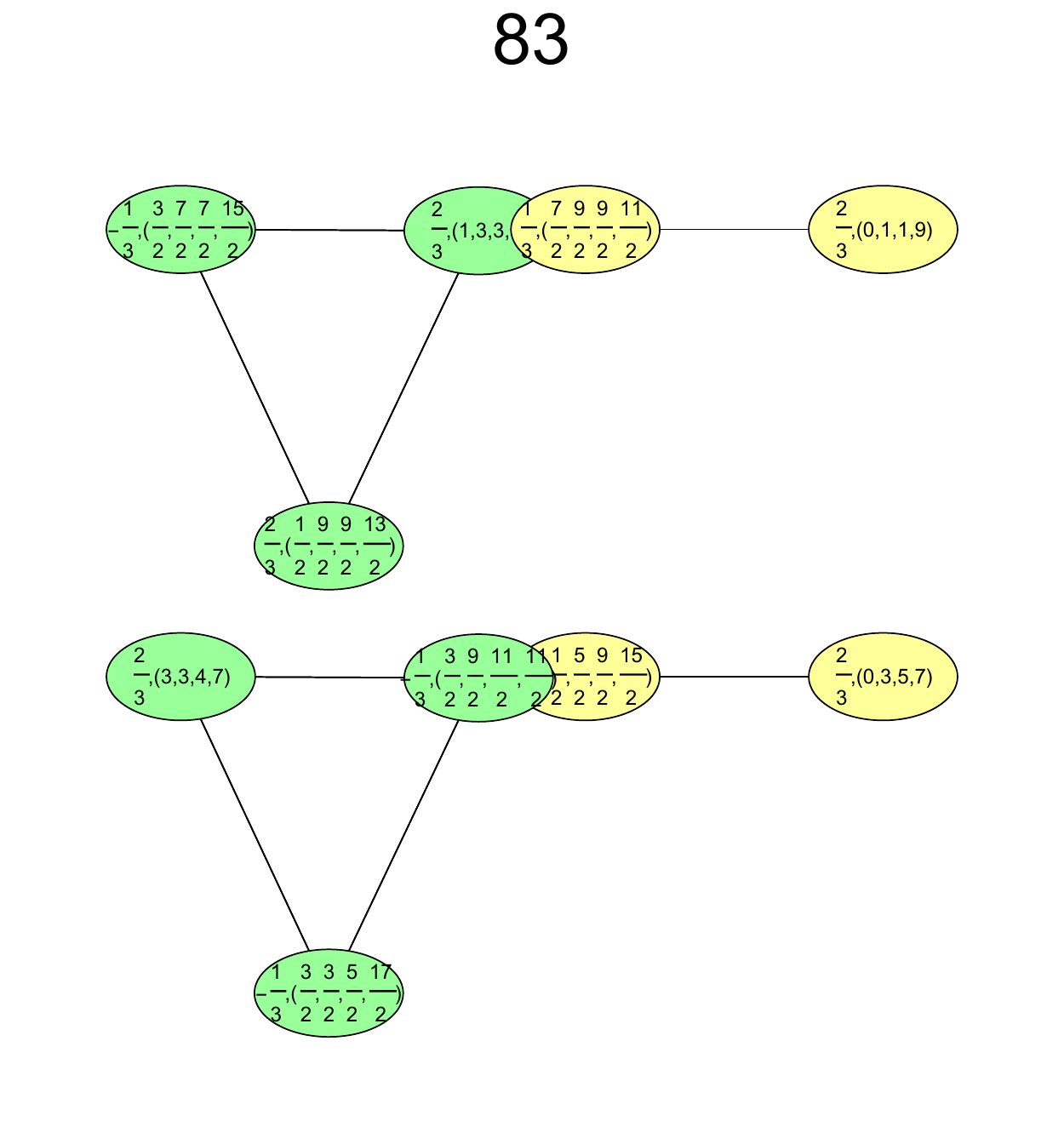}}
\caption{ \label{quarks} Hadrons with charges.  } \end{figure} 

\begin{figure}
\scalebox{0.80}{\includegraphics{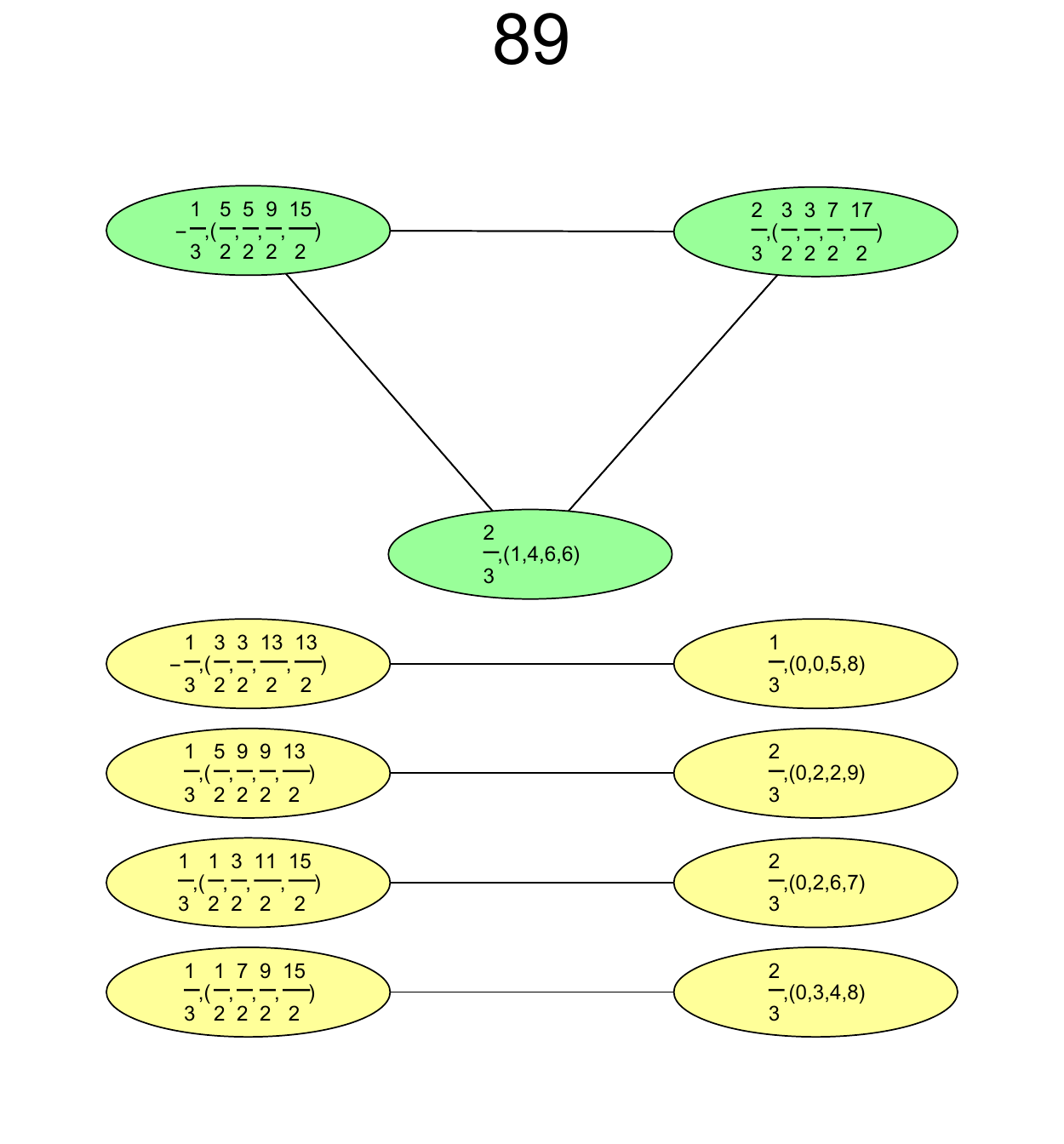}}
\scalebox{0.80}{\includegraphics{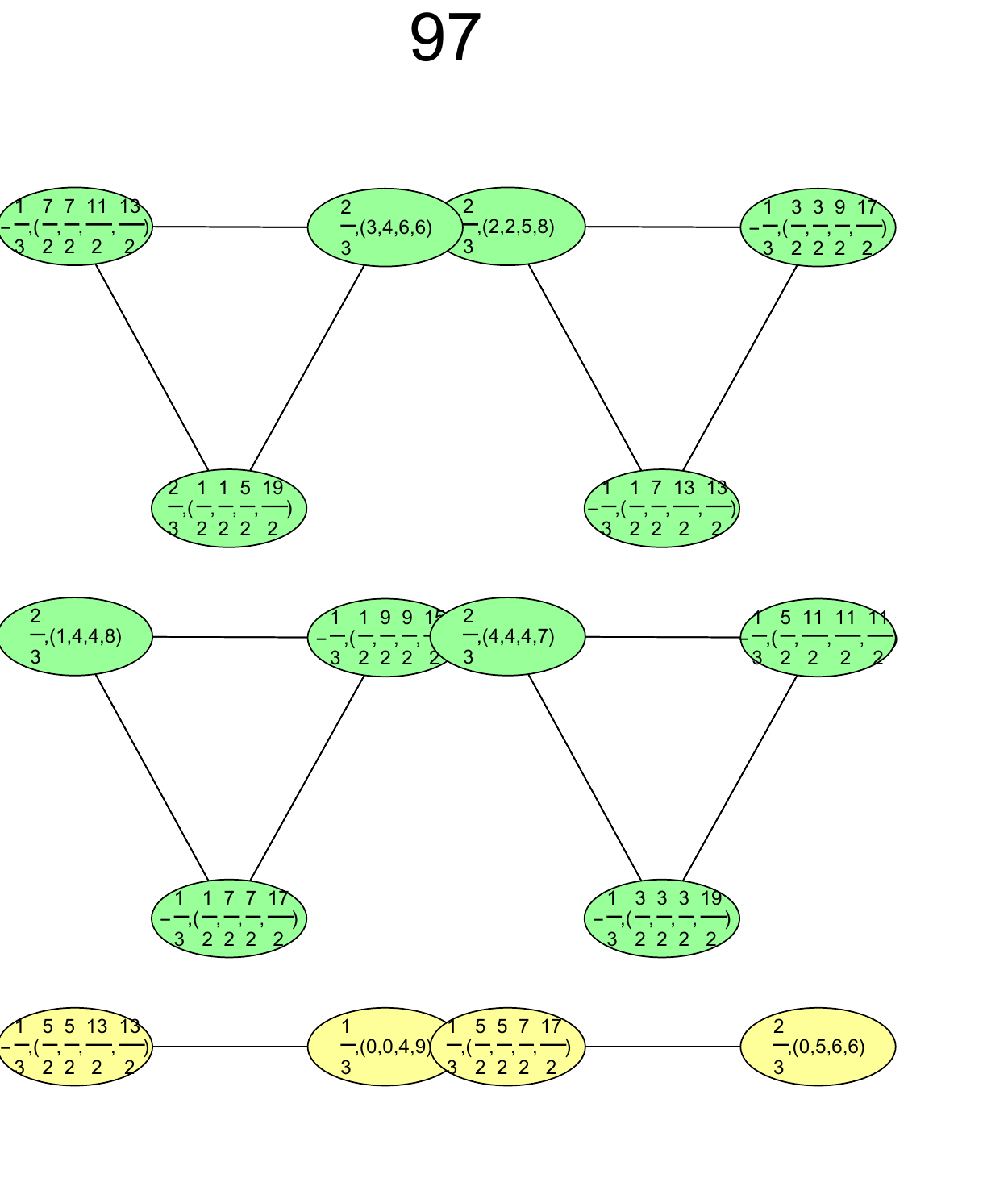}}
\caption{ \label{quarks} Hadrons with charges.  } 
\end{figure}

\begin{figure}
\scalebox{0.80}{\includegraphics{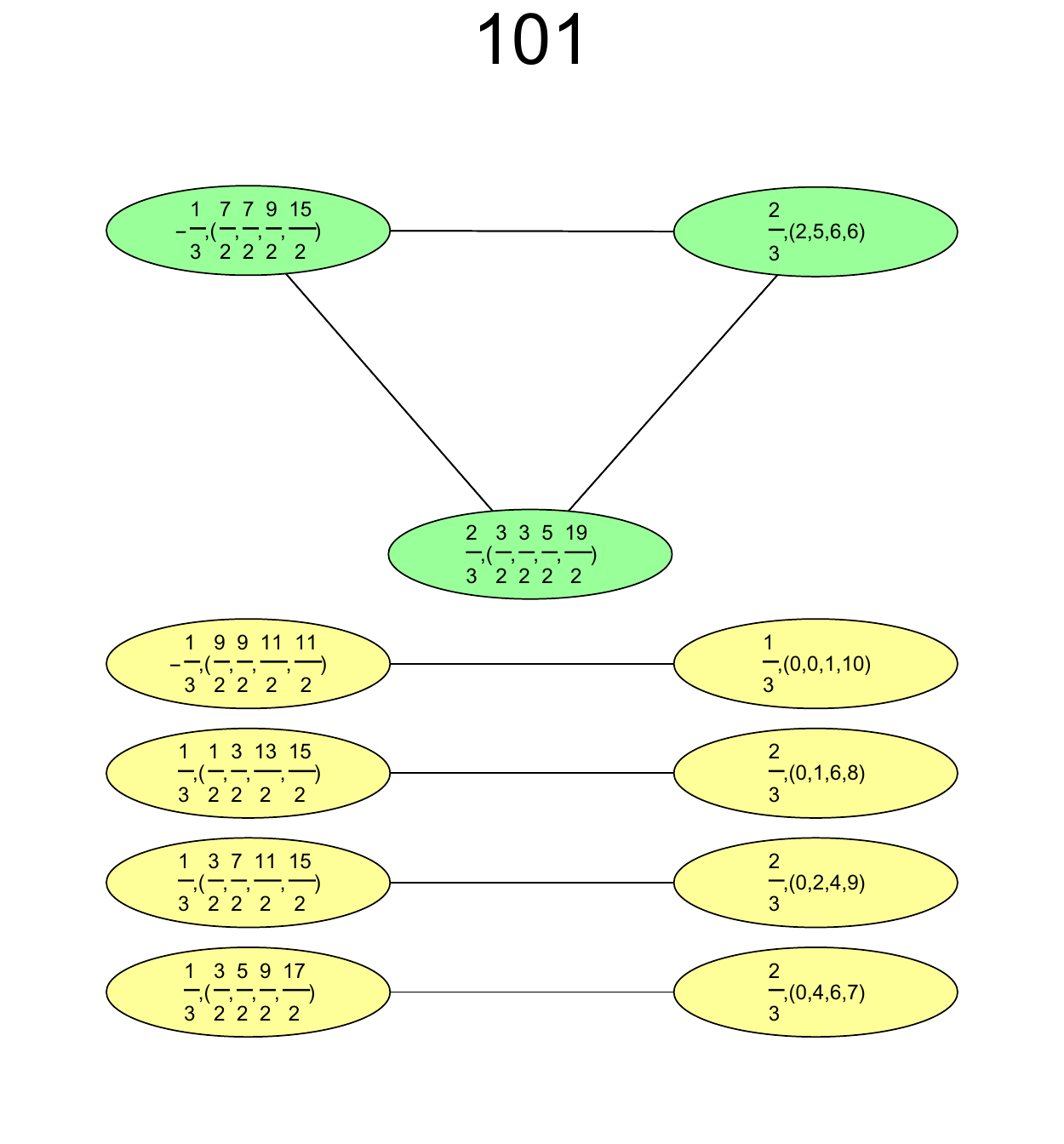}}
\scalebox{0.80}{\includegraphics{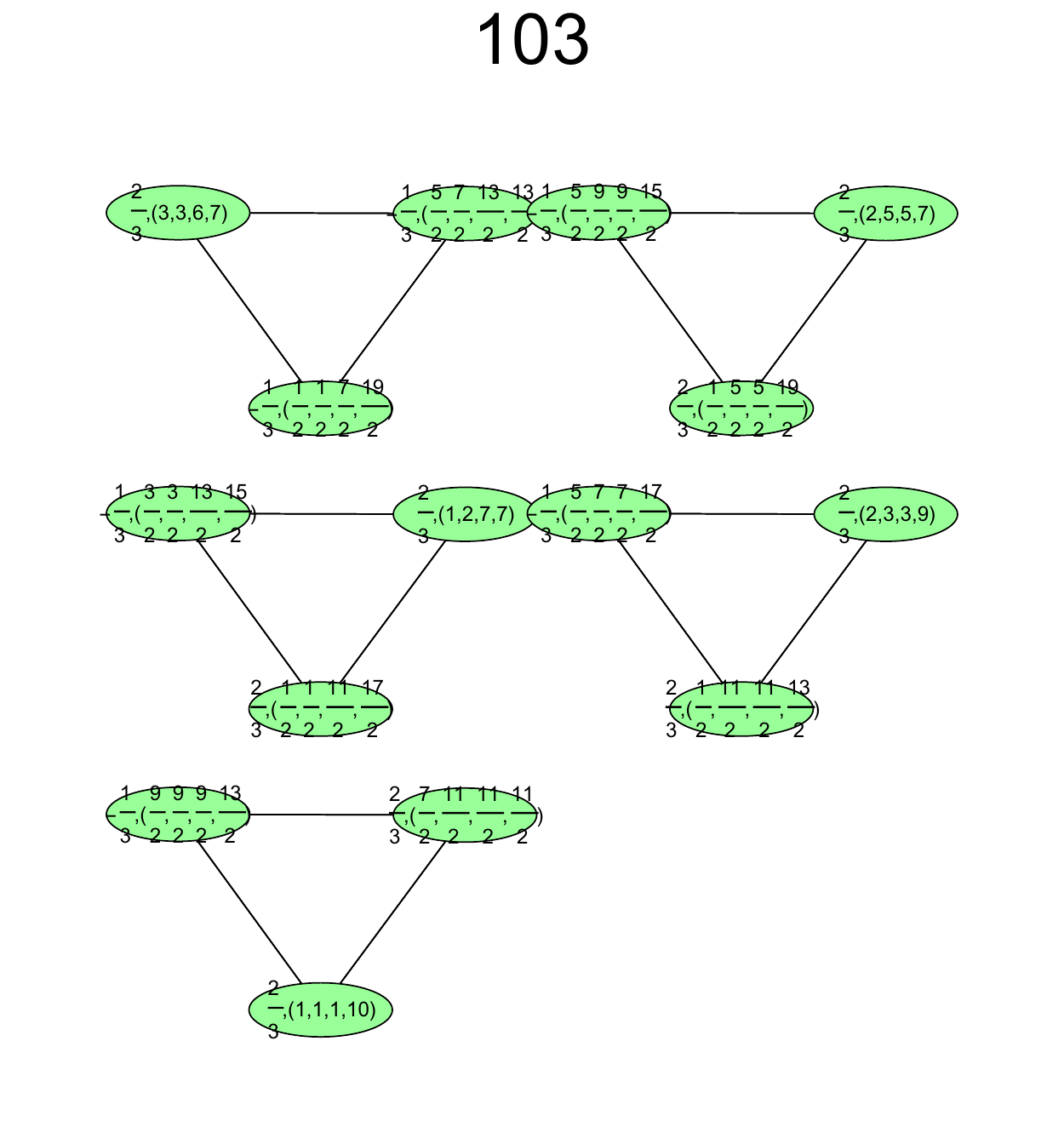}}
\caption{ \label{quarks} Hadrons with charges.  } \end{figure} 

\begin{figure}
\scalebox{0.80}{\includegraphics{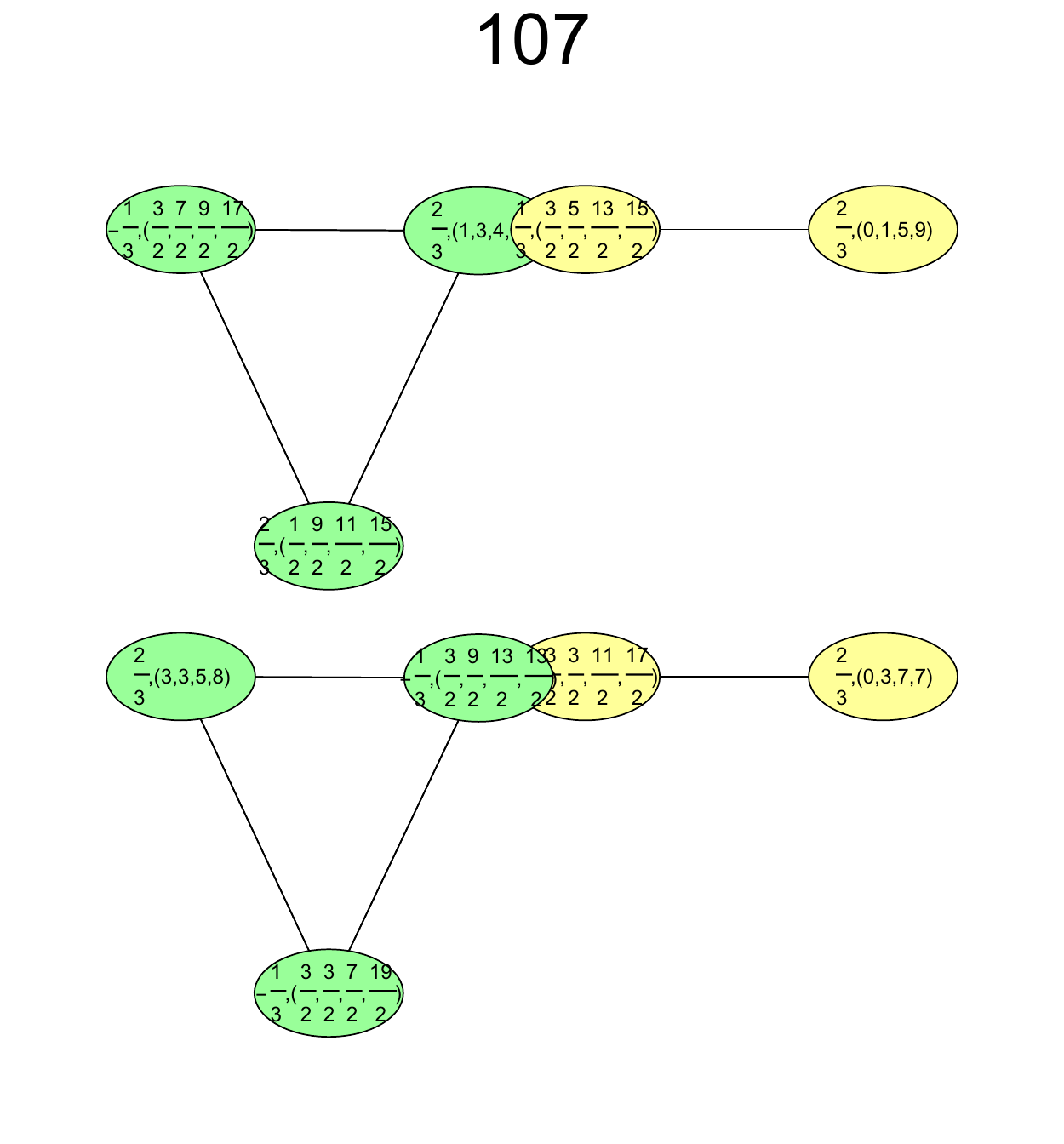}}
\scalebox{0.80}{\includegraphics{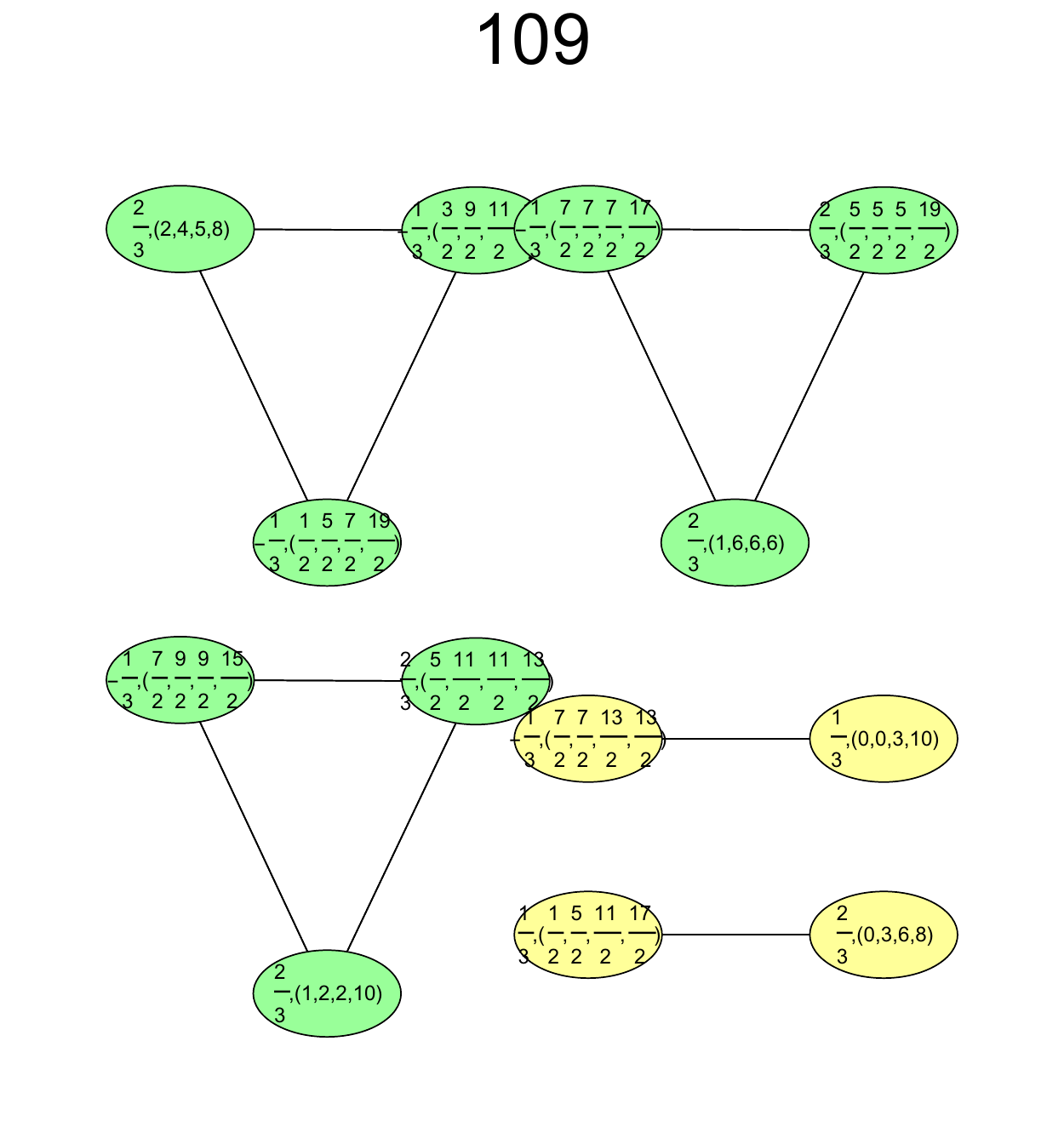}}
\caption{ \label{quarks} Hadrons with charges.  } 
\end{figure}

\pagebreak

\bibliographystyle{plain}

\begin{thebibliography}{10}

\bibitem{Baez2012}
J.C. Baez.
\newblock Division algebras and quantum theory.
\newblock {\em Found. Phys.}, 42(7):819--855, 2012.

\bibitem{CohnKumar}
H.~Cohn and A.~Kumar.
\newblock Metacommutation of {H}urwitz primes.
\newblock https://arxiv.org/abs/1307.0443, 2013.

\bibitem{ConwaySmith}
J.H. Conway and D.A. Smith.
\newblock {\em On Quaternions and Octonions}.
\newblock A.K. Peters, 2003.

\bibitem{Coxeter46}
H.S.M. Coxeter.
\newblock Integral {C}ayley numbers.
\newblock {\em Duke Math. J.}, 13:561--578, 1946.

\bibitem{Deavours73}
C.~A. Deavours.
\newblock The quaternion calculus.
\newblock {\em Amer. Math. Monthly}, 80:995--1008, 1973.

\bibitem{DeVito2010}
J.~DeVito.
\newblock Stack exchange.
\newblock http://math.stackexchange.com/questions/12453/
  is-there-an-easy-way-to-show-which-spheres-can-be-lie-groups, 2016.

\bibitem{DixonDivisionalgebras}
G.M. Dixon.
\newblock {\em Division algebras: octonions, quaternions, complex numbers and
  the algebraic design of physics}, volume 290 of {\em Mathematics and its
  Applications}.
\newblock Kluwer Academic Publishers Group, Dordrecht, 1994.

\bibitem{Numbers}
H.-D. Ebbinghaus, H.~Hermes, F.~Hirzebruch, M.~Koecher, K.~Mainzer,
  J.~Neukirch, A.~Prestel, and R.~Remmert.
\newblock {\em Numbers}, volume 123 of {\em Graduate Texts in Mathematics}.
\newblock Springer-Verlag, New York, 1991.

\bibitem{EilenbergNiven}
S.~Eilenberg and I.~Niven.
\newblock The ``fundamental theorem of algebra'' for quaternions.
\newblock {\em Bull. Amer. Math. Soc.}, 50:246--248, 1944.

\bibitem{Frobenius1879}
G.~Frobenius.
\newblock Ueber die schiefe {I}nvariante einer bilinearen oder quadratischen
  {F}orm.
\newblock {\em J. Reine Angew. Math.}, 86:44--71, 1879.

\bibitem{Fueter34}
R.~Fueter.
\newblock Die {F}unktionentheorie der {D}ifferentialgleichungen {$\Theta u=0$}
  und {$\Theta\Theta u=0$} mit vier reellen {V}ariablen.
\newblock {\em Comment. Math. Helv.}, 7(1):307--330, 1934.

\bibitem{Fueter35}
R.~Fueter.
\newblock \"{U}ber die analytische {D}arstellung der regul\"aren {F}unktionen
  einer {Q}uaternionenvariablen.
\newblock {\em Comment. Math. Helv.}, 8(1):371--378, 1935.

\bibitem{GibbsWilson}
J.W. Gibbs and E.B. Wilson.
\newblock {\em Vector Analysis}.
\newblock Yale University Press, New Haven, 1901.

\bibitem{Golomb}
S.W. Golomb.
\newblock Rubik's cube and quarks: Twists on the eight corner cells of rubik's
  cube provide a model for many aspects of quark behavior.
\newblock {\em American Scientist}, 70:257--259, 1982.

\bibitem{Hurwitz1922}
A.~Hurwitz.
\newblock \"{U}ber die {K}omposition der quadratischen {F}ormen.
\newblock {\em Math. Ann.}, 88(1-2):1--25, 1922.

\bibitem{Kirmse25}
J.~Kirmse.
\newblock {\"U}ber die {D}arstellbarkeit nat{\"u}rlicher ganzer {Z}ahlen als
  {S}ummen yon acht {Q}uadraten und {\"u}ber ein mit diesem {P}roblem
  zusammenh{\"a}ngendes nichtkommutatives und nichtassoziatives {Z}ahlensystem.
\newblock {\em Berichte Verhandlungen {S\"achs.} Akad. Wiss. Leipzig. Math.
  Phys. Kl}, 76:63--82, 1925.

\bibitem{IsospectralDirac2}
O.~Knill.
\newblock An integrable evolution equation in geometry.
\newblock {{\\} http://arxiv.org/abs/1306.0060}, 2013.

\bibitem{Experiments}
O.~Knill.
\newblock Some experiments in number theory.
\newblock {{\\} https://arxiv.org/abs/1606.05971}, 2016.

\bibitem{Kustaanheimo}
P.~Kustaanheimo.
\newblock On the fundamental prime of a finite world.
\newblock {\em Ann. Acad. Sci. Fennicae. Ser. A. I. Math.-Phys.}, 1952(129):7,
  1952.

\bibitem{Lam2001}
T.Y. Lam.
\newblock Orders and maximal orders.
\newblock https://math.berkeley.edu/~lam/html/maxorde.ps, 2001.

\bibitem{Widdows}
D.~Widdows.
\newblock {\em Quaternion algebraic geometry}.
\newblock St. Annee's College, Oxford, 2006.

\end{thebibliography}

\end{document}